\documentclass[11pt]{article}

\usepackage[margin=1in]{geometry}
\usepackage[T1]{fontenc}
\usepackage[latin9]{inputenc}
\usepackage{graphicx,verbatim}
\usepackage{amssymb}
\usepackage[ruled,vlined]{algorithm2e}

\usepackage{times,multirow}

\usepackage{url}            
\usepackage{booktabs}       
\usepackage{amsfonts}       
\usepackage{microtype}      
\usepackage{amsmath,amsbsy,amstext,amsthm}
\usepackage{color}
\usepackage{latexsym,graphics,epsf,psfrag,wrapfig}
\usepackage{booktabs}

\usepackage{bm}

\allowdisplaybreaks

\usepackage{hyperref}

\newcommand{\F}{{\mathsf F}}
\newcommand{\bA}{{\boldsymbol A}}

\newcommand{\bu}{{\boldsymbol u}}
\newcommand{\bv}{{\boldsymbol v}}
\newcommand{\ba}{{\boldsymbol a}}

\newcommand{\bh}{{\boldsymbol h}}
\newcommand{\by}{{\boldsymbol y}}
\newcommand{\bY}{{\boldsymbol Y}}
\newcommand{\bU}{{\boldsymbol U}}
\newcommand{\bX}{{\boldsymbol X}}
\newcommand{\bx}{{\boldsymbol x}}
\newcommand{\bV}{{\boldsymbol V}}

\newcommand{\bw}{{\boldsymbol w}}
\newcommand{\bI}{{\boldsymbol I}}

\newcommand{\bP}{{\boldsymbol P}}

\newcommand{\bs}{{\boldsymbol s}}
\newcommand{\bS}{{\boldsymbol S}}
\newcommand{\bR}{{\boldsymbol R}}

\newcommand{\bZ}{{\boldsymbol Z}}

\newcommand{\bQ}{{\boldsymbol Q}}
\newcommand{\bH}{{\boldsymbol H}}
\newcommand{\bM}{{\boldsymbol M}}
\newcommand{\bW}{{\boldsymbol W}}
\newcommand{\bT}{{\boldsymbol T}}

\newcommand{\argmin}{\mathrm{argmin}}

\theoremstyle{plain}\newtheorem{lemma}{\textbf{Lemma}}\newtheorem{theorem}{\textbf{Theorem}}\newtheorem{corollary}{\textbf{Corollary}}

\theoremstyle{definition}

\theoremstyle{remark}\newtheorem{remark}{\textbf{Remark}}

\begin{document}

\title{Nonconvex Matrix Factorization from Rank-One Measurements}



\author{Yuanxin Li\thanks{Department of Electrical and Computer Engineering, Carnegie Mellon University, Pittsburgh, PA 15213, USA; Email: \href{mailto:yuanxinl@andrew.cmu.edu}{yuanxinl@andrew.cmu.edu}} \and Cong Ma\thanks{Department of Operations Research and Financial Engineering, Princeton University, Princeton, NJ 08544, USA; Email: \href{mailto:congm@princeton.edu}{congm@princeton.edu}} \and Yuxin Chen\thanks{Department of Electrical Engineering, Princeton University, Princeton, NJ 08544, USA; Email: \href{mailto:yuxin.chen@princeton.edu}{yuxin.chen@princeton.edu}} \and Yuejie Chi\thanks{Department of Electrical and Computer Engineering, Carnegie Mellon University, Pittsburgh, PA 15213, USA; Email: \href{mailto:yuejiechi@cmu.edu}{yuejiechi@cmu.edu}}}

\date{\today}

\maketitle

\begin{abstract}
We consider the problem of recovering low-rank matrices from random rank-one measurements, which spans numerous applications including covariance sketching, phase retrieval, quantum state tomography, and learning shallow polynomial neural networks, among others. Our approach is to directly estimate the low-rank factor by minimizing a nonconvex quadratic loss function via vanilla gradient descent, following a tailored spectral initialization. When the true rank is small, this algorithm is guaranteed to converge to the ground truth (up to global ambiguity) with near-optimal sample complexity and computational complexity. To the best of our knowledge, this is the first guarantee that achieves near-optimality in both metrics. In particular, the key enabler of near-optimal computational guarantees is an implicit regularization phenomenon: without explicit regularization, both spectral initialization and the gradient descent iterates automatically stay within a region incoherent with the measurement vectors. This feature allows one to employ much more aggressive step sizes compared with the ones suggested in prior literature, without the need of sample splitting.
\end{abstract}



\section{Introduction}

This paper is concerned with estimating a low-rank positive semidefinite matrix $\bM^{\natural}\in \mathbb{R}^{n\times n}$ from a few {\em rank-one measurements}. Specifically, suppose that the matrix of interest can be factorized as $$\bM^{\natural}=\bX^{\natural}\bX^{\natural\top} \in \mathbb{R}^{n\times n},$$ 
where 
$\bX^{\natural}\in\mathbb{R}^{n\times r}$ ($r\ll n$) denotes the low-rank factor. We collect $m$ measurements $\{y_i\}_{i=1}^{m}$ about $\bM^{\natural}$ taking the form 
\begin{equation*}
y_i = \ba_i^{\top} \bM^{\natural} \ba_i = \big\| \ba_i^{\top}\bX^{\natural} \big\|_2^2, \quad i = 1,\cdots,m,
\end{equation*}
where $\left\{\ba_{i} \in \mathbb{R}^{n}\right\}_{i=1}^{m}$ represent the measurement vectors known {\em a priori}. 
One can think of $\{\ba_i \ba_i^{\top}\}_{i=1}^{m}$ as a set of linear sensing matrices (so that $y_i = \langle \ba_i \ba_i^{\top}, \bM^{\natural} \rangle$), which are all rank-one\footnote{Given that $y_i$ is a quadratic function with respect to both $\bX^{\natural}$ and $\ba_i$, the measurement scheme is also referred to as {\em quadratic sampling}.}.  The goal is to recover $\bM^{\natural}$, or equivalently, the low-rank factor $\bX^{\natural}$, from a limited number of rank-one measurements. 
This problem spans a variety of important practical applications, with a few examples listed below.
\begin{itemize}
	\item \textbf{Covariance sketching.}  Consider a zero-mean data stream $\{\bx_t\}_{t\in \mathcal{T}}$, whose covariance matrix  $\bM^{\natural}:=\mathbb{E}[\bx_t\bx_t^{\top}]$ is (approximately) low-rank. To estimate the covariance matrix, one can collect $m$ aggregated quadratic sketches of the form 
\begin{equation*}
y_i = \frac{1}{|\mathcal{T}|}\sum_{t\in \mathcal{T}} (\ba_i^{\top} \bx_t)^2, 
\end{equation*}
which converges to  $\mathbb{E}[(\ba_i^{\top}\bx_t)^2]= \ba_i^{\top} \bM^{\natural} \ba_i$ as the number of data instances grows. This quadratic covariance sketching scheme can be performed under minimal storage requirement and low sketching cost. See \cite{chen2013exact} for detailed descriptions.

	%

	\item \textbf{Phase retrieval and mixed linear regression.} This problem  subsumes as a special case the phase retrieval problem \cite{candes2015phase}, which aims to estimate an unknown signal $\bx^{\natural} \in \mathbb{R}^n$ from intensity measurements (which can often be modeled or approximated by quadratic measurements of the form $y_i=(\ba_i^{\top}\bx^{\natural})^2$). 
		This problem has found numerous applications in X-ray crystallography, optical imaging, astronomy, etc.  Another related problem in machine learning is mixed linear regression with two components, where the data one collects are generated from one of two unknown regressors; see \cite{chen2014convex} for precise formulation.

\item \textbf{Quantum state tomography.} Estimating the density operator of a quantum system can be formulated as a low-rank positive semidefinite matrix recovery problem using rank-one measurements, when the density operator is {\em almost pure} \cite{kueng2017low}. A problem of similar mathematical formulation occurs in phase space tomography \cite{tian2012experimental}, where the goal is to reconstruct the correlation function of a wave field.
\item \textbf{Learning shallow polynomial neural networks.} Taking $\{\ba_i,y_i\}_{i=1}^m$ as training data, our problem is equivalent to learning a one-hidden-layer, fully-connected neural network with a quadratic activation function \cite{livni2014computational,soltanolkotabi2017theoretical,soltani2018towards}, where the output of the network is expressed as $y = \sum_{i=1}^r  \sigma(\ba^{\top}\bx_i^{\natural})$ with $\bX^{\natural}=[\bx_1^{\natural}, \bx_2^{\natural}, \cdots,\bx_r^{\natural}] \in \mathbb{R}^{n\times r}$ and the activation function $\sigma(z) =z^2$. 
\end{itemize}

\subsection{Main Contributions}

Due to the quadratic nature of the measurements,  the natural least-squares empirical risk  formulation is highly nonconvex and in general challenging to solve. 
To be more specific, consider the following optimization problem:
\begin{equation}\label{eq:loss}
\text{minimize}_{\bX\in\mathbb{R}^{n\times r}} \quad f\left( \bX \right) : = \frac{1}{4m} \sum_{i=1}^{m} \left(y_{i} - \big\Vert \ba_{i}^{\top}\bX \big\Vert_{2}^{2}\right)^{2},
\end{equation}
which aims to optimize a degree-4 polynomial in $\bX$ and is NP hard in general. 
The problem, however, may become tractable under certain random designs, and may even be solvable using simple methods like gradient descent. 
Our main finding is the following: under i.i.d.~Gaussian design (i.e.~$\ba_i\sim \mathcal{N}(\bm{0},\bI_{n})$), {\em vanilla gradient descent} combined with spectral initialization  achieves appealing performance guarantees both statistically and computationally.


\begin{itemize}
	\item  Statistically, we show that gradient descent converges exactly to the true factor $\bX^{\natural}$ (modulo unrecoverable global ambiguity), as soon as the number of measurements exceeds the order of $O(nr^4\log n)$. 
	       When $r$ is fixed independent of $n$, this sample complexity is near-optimal up to some logarithmic factor.   		
		
        \item Computationally, to achieve $\epsilon$-accuracy, gradient descent requires an iteration complexity of $O(r^2  \log(1/\epsilon))$ (up to logarithmic factors), with a per-iteration cost of $O(mnr)$. 
	       When $r$ is fixed independent of $m$ and $n$, the computational complexity scales linearly with $mn$, which is proportional to the time taken to read all data.   	       


\end{itemize}

These findings significantly improve upon existing results that  require either resampling (which is not sample-efficient and is not the algorithm one actually runs in practice \cite{zhong2015efficient,lin2016non,soltani2018towards}), or high iteration complexity (which results in high computation cost \cite{sanghavi2017local}). In particular, our work is most related to \cite{sanghavi2017local} that also studied the effectiveness of gradient descent. The results in \cite{sanghavi2017local} require a sample complexity on the order of $nr^6\log^{2}{n}$,  as well as an iteration complexity of $O(n^{4}r^{2} \log(1/\epsilon))$ (up to logarithmic factors) to attain $\epsilon$-accuracy. In comparison, our theory  improves the sample complexity to $O(nr^4\log n)$ and, perhaps more importantly, establishes a much lower iteration complexity  of $O( r^{2} \log(1/\epsilon))$ (up to logarithmic factor). To the best of our knowledge, this work is the first nonconvex algorithm (without resampling) that achieves  both near-optimal  statistical and computational guarantees with respect to $n$. 


\subsection{Surprising Effectiveness of Gradient Descent}


Recently, gradient descent has been widely employed to address various nonconvex optimization problems due to its appealing efficiency from both statistical and computational perspectives. Despite the nonconvexity of \eqref{eq:loss}, \cite{sanghavi2017local} showed that within a local neighborhood of $\bX^{\natural}$, where $\bX$ satisfies
\begin{align}
	\big\Vert\bX - \bX^{\natural}\big\Vert_{\F} & \le \frac{1}{24}\frac{\sigma_{r}^{2}\left(\bX^{\natural}\right)}{\left\Vert \bX^{\natural} \right\Vert_{\F}}, \label{eq:local-region}
\end{align}
$f(\bX)$ behaves like a strongly convex function, at least along certain descending directions. However, this region itself is not enough to guarantee computational efficiency, and consequently, the smoothness parameter derived in \cite{sanghavi2017local} is as large as $n^2$ (even ignoring additional polynomial factors in $r$), leading to a step size as small as $O(1/n^4)$ and an iteration complexity of $O(n^4 \log(1/\epsilon))$. These are fairly pessimistic. 

In order to improve computational guarantees, it might be tempting to employ appropriately designed regularization operations --- such as truncation \cite{ChenCandes15solving} and projection \cite{chen2015fast}. These explicit regularization operations  are capable of stabilizing the search direction, and make sure the whole trajectory is in a basin of attraction with benign curvatures surrounding the ground truth. However, such explicit regularizations complicate algorithm implementations, as they introduce more tuning parameters.

Our work is inspired by \cite{ma2017implicit}, which uncovers the ``implicit regularization'' phenomenon of vanilla gradient descent for nonconvex estimation problems such as phase retrieval and low-rank matrix completion. In words, even without extra regularization operations,  vanilla gradient descent always follows a path within some region around the global optimum with nice geometric structure, at least along certain directions. The current paper demonstrates that a similar phenomenon persists in low-rank matrix factorization from rank-one measurements. 

To describe this phenomenon in a precise manner, we need to specify which region enjoys the desired geometric properties. To this end, 
consider a local region around $\bX^{\natural}$ where $\bX$ is ``incoherent''\footnote{This is called incoherent because if $\bX$ is aligned (and hence coherent) with the sensing vectors,  $\big\Vert \ba_{l}^{\top}\big(\bX - \bX^{\natural}\big) \big\Vert_{2}$ can be $O(\sqrt{n})$ times larger than the right-hand side of \eqref{eq:incoherent}.} with all  sensing vectors in the following sense:
\begin{equation}
	\label{eq:incoherent} 
	\max_{1\le l\le m} \big\Vert \ba_{l}^{\top}\big(\bX - \bX^{\natural}\big) \big\Vert_{2}  \le \frac{1}{24} \sqrt{\log{n}}  \cdot \frac{\sigma_{r}^{2}\left(\bX^{\natural}\right)}{\Vert \bX^{\natural} \Vert_{\F}}. 
\end{equation}
  We term the intersection of \eqref{eq:local-region} and \eqref{eq:incoherent} the {\em Region of Incoherence and Contraction} (RIC). 
     The nice feature of the RIC is this: within this region, 
  the loss function $f(\bX)$ enjoys a smoothness parameter that scales as $O(\max\{r, \log n\})$ (namely, $\|\nabla^2 f(\bx) \| \lesssim \max\{r, \log n\}$, which is much smaller than $O(n^2)$ provided in \cite{sanghavi2017local}).  As is well known,  a region enjoying a smaller smoothness parameter  enables more aggressive progression of gradient descent. 
  

A key question remains as to how to prove that the trajectory of gradient descent never leaves the RIC. This is, unfortunately, not guaranteed by standard optimization theory, which only ensures contraction of the Euclidean error. To address this issue, we resort to the leave-one-out trick \cite{ma2017implicit,zhong2017near,chen2017spectral} that produces  auxiliary trajectories of gradient descent that use all but one sample.  This allows us to establish the incoherence condition by leveraging the statistical independence of the leave-one-out trajectory w.r.t.~the corresponding sensing vector that has been left out. Our theory refines the leave-one-out argument and further establishes linear contraction in terms of the entry-wise prediction error. 

\subsection{Notations}
We use boldface lowercase (resp.~uppercase) letters to represent vectors (resp.~matrices). We denote by $\left\Vert\bx\right\Vert_{2}$ the $\ell_2$ norm of a vector $\bx$, and  $\bX^{\top}$, $\left\Vert\bX\right\Vert$ and $\left\Vert\bX\right\Vert_{\F}$  the transpose, the spectral norm and the Frobenius norm of a matrix $\bX$, respectively. The $k$th largest singular value of a matrix $\bX$ is denoted by $\sigma_{k}\left(\bX\right)$. Moreover, the inner product between two matrices $\bX$ and $\bY$ is defined as $\langle \bX, \bY\rangle = \mathrm{Tr}\left(\bY^{\top}\bX\right)$, where $\mathrm{Tr}\left(\cdot\right)$ is the trace. We also use $\mathrm{vec}(\bV)$ to denote vectorization of a matrix $\bV$. 
The notation $f(n) \lesssim g(n)$ or $f(n)=O(g(n))$ means that  there exists a universal constant $c >
0$ such that $|f(n)| \leq c|g(n)|$.
In addition, we use $c$ and $C$ with different subscripts to represent positive numerical constants, whose values may change from line to line.




\section{Algorithms and Main Results}\label{sec_problem_formulation}

To begin with, we present the formal problem setup. Suppose we are given a set of $m$ rank-one measurements as follows
\begin{equation}
y_{i} =  \big\Vert \ba_{i}^{\top}\bX^{\natural} \big\Vert_{2}^{2}, \qquad i = 1,\cdots,m,
\end{equation}
where $\ba_{i}\in\mathbb{R}^{n}$ is the $i$th sensing vector composed of i.i.d.~standard Gaussian entries, i.e.~$\ba_{i} \sim \mathcal{N}\left(\boldsymbol{0}, \bI_n\right)$, for $i=1,\cdots,m$. The underlying ground truth $\bX^{\natural}\in\mathbb{R}^{n\times r}$ is assumed to have full column rank but not necessarily having orthogonal columns. Define the condition number of $\bM^{\natural} = \bX^{\natural}\bX^{\natural\top}$ as
\begin{equation}
\kappa = \frac{\sigma_{1}^{2}\left(\bX^{\natural}\right)}{\sigma_{r}^{2}\left(\bX^{\natural}\right)}.
\end{equation} 
Throughout this paper, we assume the condition number is bounded by some constant independent of $n$ and $r$, i.e.~$\kappa = O(1)$. Our goal is to recover $\bX^{\natural}$, up to (unrecoverable) orthonormal transformation, from the measurements $\by=\left\{y_{i}\right\}_{i=1}^{m}$ in a  statistically and computationally efficient manner.

\subsection{Vanilla Gradient Descent}
 
The algorithm studied herein  is a combination of vanilla gradient descent and a judiciously designed spectral initialization.  Specifically, consider minimizing the squared loss:
\begin{equation}
 f\left( \bX \right) : = \frac{1}{4m} \sum_{i=1}^{m} \left(y_{i} - \big\Vert \ba_{i}^{\top}\bX\big\Vert_{2}^{2}\right)^{2},
\end{equation}
which is a nonconvex function. We attempt to optimize this function iteratively via gradient descent
\begin{equation}
\bX_{t+1} = \bX_{t} - \mu_{t} \nabla f\left(\bX_{t}\right), \qquad t = 0,1,\cdots,
\end{equation} 
where $\bX_{t}$ denotes the estimate in the $t$th iteration, $\mu_{t}$ is the step size/learning rate, and the gradient $\nabla f(\bX)$ is given by
 \begin{equation}
\nabla f\left(\bX\right) =  \frac{1}{m} \sum_{i=1}^{m } \left( \big\Vert \ba_{i}^{\top}\bX \big\Vert_{2}^{2} - y_{i} \right) \ba_{i}\ba_{i}^{\top}\bX.
\end{equation}

For initialization, similar to \cite{sanghavi2017local},\footnote{Compared with \cite{sanghavi2017local}, when setting the eigenvalues in \eqref{equ_initial_eigenvalue}, we use the sample mean $\lambda$ rather than $\lambda_{r+1}\left(\bY\right)$ to estimate $\frac{1}{2} \Vert \bX^{\natural} \Vert_{\F}^{2}$.} we apply the spectral method, which sets the columns of $\bX_0$ as the top-$r$ eigenvectors --- properly scaled --- of a matrix $\bY$ as defined in \eqref{def_Y}. The rationale is this: the mean of $\bY$ is given by
$$\mathbb{E}\left[\bY\right] =  \frac{1}{2} \big\Vert \bX^{\natural}\big\Vert_{\F}^{2}\,\bI_{n} + \bX^{\natural}\bX^{\natural \top}, $$
and hence the principal components of $\bY$ form a reasonable estimate of $\bX^{\natural}$, provided that there are sufficiently many samples.  The full algorithm is described in Algorithm~\ref{alg:ncvx_sketching}.



\begin{algorithm}[t]

\caption{Gradient Descent with Spectral Initialization}
\label{alg:ncvx_sketching}
\noindent \textbf{Input:} measurements $\by = \left\{y_{i}\right\}_{i=1}^{m}$, and sensing vectors $\left\{\ba_{i}\right\}_{i=1}^{m}$.

\noindent \textbf{Parameters:} step size $\mu_{t}$, rank $r$, and number of iterations $T$.

\noindent \textbf{Initialization:} set $\bX_{0} = \bZ_{0}\boldsymbol{\Lambda}_{0}^{1/2}$, where the columns of $\bZ_{0} \in\mathbb{R}^{n\times r}$ contain the normalized eigenvectors corresponding to the $r$ largest eigenvalues of the matrix 
\begin{equation}\label{def_Y}
\bY = \frac{1}{2m} \sum_{i=1}^{m} y_{i}\ba_{i}\ba_{i}^{\top},
\end{equation}
and $\boldsymbol{\Lambda}_{0}$ is an $r\times r$ diagonal matrix, with the entries on the diagonal given as 
\begin{equation}\label{equ_initial_eigenvalue}
\left[\boldsymbol{\Lambda}_{0}\right]_{i} = \lambda_{i}\big(\bY\big) - \lambda, \quad  \ i=1,\cdots,r,
\end{equation}
where $\lambda = \frac{1}{2m} \sum_{i=1}^{m} y_{i}$ and $\lambda_{i}\left(\bY\right)$ is the $i$th largest eigenvalue of $\bY$.

\noindent \textbf{Gradient loop:} for $t = 0:1:T-1$, do 
\begin{equation}
\bX_{t+1} = \bX_{t} - \mu_{t} \cdot \frac{1}{m} \sum_{i=1}^{m } \left( \big\Vert \ba_{i}^{\top}\bX_{t} \big\Vert_{2}^{2} - y_{i} \right) \ba_{i}\ba_{i}^{\top}\bX_{t}.
\end{equation}

\noindent \textbf{Output:} $\bX_{T}$.

\end{algorithm}


\subsection{Performance Guarantees}

Before proceeding to our main results, we specify the metric used to assess the estimation error of the running iterates. Since $\left(\bX^{\natural} \bP\right)\left(\bX^{\natural} \bP\right)^{\top} = \bX^{\natural}\bX^{\natural \top}$ for any orthonormal matrix $\bP\in\mathbb{R}^{r\times r}$, $\bX^{\natural}$ is recoverable up to orthonormal transforms. Hence, we define the error of the $t$th iterate $\bX_t$ as
\begin{equation}
\mathrm{dist} \big(\bX_{t}, \bX^{\natural}\big) = \big\Vert \bX_{t}\bQ_{t} - \bX^{\natural}\big\Vert_{\F},
\end{equation}
where $\bQ_t$ is given by
\begin{equation}
\bQ_{t} := \argmin_{\bP\in\mathcal{O}^{r\times r}} \big\Vert \bX_{t}\bP - \bX^{\natural}\big\Vert_{\F}
	\label{eq:defn-Qt}
\end{equation}
with $\mathcal{O}^{r\times r}$ denoting the set of all $r \times r$ orthonormal matrices. Accordingly, we have the following theoretical performance guarantees of Algorithm~\ref{alg:ncvx_sketching}.

\begin{theorem}
\label{main_theorem}
Suppose that $m\ge c  n r^{3} (r + \sqrt{\kappa}) \kappa^{3} \log{ n}$ with some large enough constant $c>0$, and that the step size obeys $0<\mu_{t}:= \mu =  \frac{c_{4}}{( r\kappa + \log{n})^{2} \sigma_r^2(\bX^{\natural})}$. Then with probability at least $1 - O(mn^{-7})$, the iterates satisfy
\begin{equation}\label{equ_main_theorem_estbound}
\mathrm{dist} \big(\bX_{t}, \bX^{\natural}\big) \leq c_1 \left( 1-  0.5\mu \sigma_r^2(\bX^{\natural})  \right)^t    \frac{\sigma_{r}^{2}\left(\bX^{\natural}\right)}{\left\Vert\bX^{\natural}\right\Vert_{\F}}, 
\end{equation}
for all $t \ge 0$. In addition, 
\begin{equation}\label{equ_main_theorem_incoherence}
\max_{1\le l\le m} \left\Vert \ba_{l}^{\top} \big( \bX_{t}\bQ_{t} - \bX^{\natural} \big) \right\Vert_{2} \le c_{2} \left( 1-  0.5\mu \sigma_r^2(\bX^{\natural})  \right)^t \sqrt{\log{n}} \cdot \frac{\sigma_{r}^{2}\left(\bX^{\natural}\right)}{\left\Vert\bX^{\natural}\right\Vert_{\F}}, 
\end{equation}
for all $0 \le t \le c_{3}n^{5}$. Here, $c_{1},\cdots,c_{4}$ are some universal positive constants.
\end{theorem}

\begin{remark}
The precise expression of required sample complexity in Theorem~\ref{main_theorem} can be written as $m\ge c   \max\left\{ \frac{\left\Vert \bX^{\natural} \right\Vert_{\F}}{\sigma_{r}\left( \bX^{\natural} \right) } \sqrt{r},  \kappa  \right\} \frac{\left\Vert \bX^{\natural} \right\Vert_{\F}^{5}}{\sigma_{r}^{5}\left( \bX^{\natural} \right) }  n \sqrt{r}  \log{\left(\kappa n\right)}$ with some large enough constant $c>0$. By adjusting constants, with probability at least $1 - O(mn^{-7})$,  \eqref{equ_main_theorem_incoherence} holds for $0\le t \le O(n^{c_{5}})$ in any power $c_{5}\geq 1$.
\end{remark}

Theorem~\ref{main_theorem} has the following implications.
\begin{itemize}
	\item {\em Near-optimal sample complexity when $r$ is fixed:} Theorem~\ref{main_theorem} suggests that spectrally-initialized vanilla gradient descent succeeds as soon as $m=O(nr^4\log n)$. When $r=O(1)$,  this leads to near-optimal sample complexity up to logarithmic factor. In fact, once the spectral initialization is finished, a sample complexity at $m=O(nr^3\log n)$ can guarantee the linear convergence to the global optima. To the best of our knowledge, this outperforms all performance guarantees in the literature obtained for any nonconvex method without requiring {\em resampling}. 

	\item {\em Near-optimal computational complexity:} In order to achieve $\epsilon$-accuracy, i.e.~$\mathrm{dist} \left(\bX_{t}, \bX^{\natural}\right) \leq \epsilon \|\bX\|_{\F}$, it suffices to run gradient descent for $T= O\left(  r^{2} \mathrm{poly}\log(n) \log(1/\epsilon) \right)$ iterations. This results in a total computational complexity of $O\left(mnr^3  \mathrm{poly}\log(n) \log(1/\epsilon) \right)$. 

\item {\em Implicit regularization:} Theorem~\ref{main_theorem} demonstrates that both the spectral initialization and the gradient descent updates provably control the entry-wise error $\max_{1\le l\le m} \left\Vert \ba_{l}^{\top} \big( \bX_{t}\bQ_{t} - \bX^{\natural} \big) \right\Vert_{2}$, and the iterates remain incoherent with respect to all the sensing vectors. In fact, the entry-wise error decreases linearly as well, which is not characterized in \cite{ma2017implicit}.  
 
\end{itemize}

Theorem~\ref{main_theorem} is established using a fixed step size. According to our theoretical analysis, the incoherence condition \eqref{equ_main_theorem_incoherence} has a significant impact on the convergence rate. After a few iterations, the incoherence condition can be bounded independent of $\log n$, which leads to a larger step size and faster convergence. Specifically, we have the following corollary.

\begin{corollary}
Under the same setting of Theorem~\ref{main_theorem}, after $T_{a} = c_{6} \max\{\kappa^{2}r^{2}\log{n}, \log^{3}{n}\}$ iterations, the step size can be relaxed as $0 < \mu_{t}:=\mu =  \frac{c_{7}}{ r^{2} \kappa^{2}  \sigma_r^2(\bX^{\natural})} $, with some universal constant $c_{6}, c_{7}>0$, then the iterates satisfy
\begin{equation}
\mathrm{dist} \big(\bX_{t}, \bX^{\natural}\big) \leq c_{1} \left( 1-  0.5 \mu \sigma_r^2(\bX^{\natural})  \right)^t \frac{\sigma_{r}^{2}\left(\bX^{\natural}\right)}{\left\Vert\bX^{\natural}\right\Vert_{\F}}, 
\end{equation}
for all $t \ge T_{a}$, with probability at least $1 - O(mn^{-7})$.
\end{corollary}



\section{Related Work}
 
Instead of directly estimating $\bX^{\natural}$, the problem of interest can be also solved by estimating $\bM^{\natural} = \bX^{\natural}\bX^{\natural\top}$ in higher dimension via nuclear norm minimization, which requires $O(nr)$ measurements for exact recovery \cite{chen2013exact,cai2015rop,kueng2017low,li2017low}. See also \cite{candes2013phaselift,candes2014solving, demanet2014stable,waldspurger2015phase} for the phase retrieval problem. However, nuclear norm minimization, often cast as the semidefinite programming, is in general computationally expensive to deal with large-scale data.

On the other hand, nonconvex approaches have drawn intense attention in the past decade due to their ability to achieve computational and statistical efficiency all at once. Specifically, for the phase retrieval problem, Wirtinger Flow (WF) and its variants \cite{candes2015phase, ChenCandes15solving, cai2016optimal, ma2017implicit, zhang2017reshaped, soltanolkotabi2017structured, wang2017solving} have been proposed. As a two-stage algorithm, it consists of spectral initialization and iterative gradient updates. This strategy has found enormous success in solving other problems such as low-rank matrix recovery and completion \cite{chen2015fast,tu2016low}, blind deconvolution \cite{DBLP:journals/corr/LiLSW16}, and spectral compressed sensing \cite{cai2017spectral}. We follow a similar route but analyze a more general problem that includes phase retrieval as a special case.

The paper \cite{sanghavi2017local} is most close to our work, which studied the local convexity of the same loss function and developed performance guarantees for gradient descent using a similar, but different spectral initialization scheme. As discussed earlier, due to the pessimistic estimate of the smoothness parameter, they only allow a diminishing learning rate (or step size) of $O(1/n^4)$, leading to a high iteration complexity. We not only provide stronger computational guarantees, but also improve the sample complexity, compared with \cite{sanghavi2017local}. 
\begin{table}[ht]
\begin{center}
\begin{tabular}{|c|c|c|}  
\hline 
Algorithms with resampling & Sample complexity & Computational complexity   \tabularnewline
\hline \hline
AltMin-LRROM  \cite{zhong2015efficient}  & $O(nr^4\log^2{n}\log{(\frac{1}{\epsilon})})$ & $O(mnr\log{(\frac{1}{\epsilon})})$    \tabularnewline
\hline
gFM  \cite{lin2016non}  & $O(nr^3\log{(\frac{1}{\epsilon})})$ & $O(mnr\log{(\frac{1}{\epsilon})})$    \tabularnewline
\hline 
EP-ROM  \cite{soltani2018towards}& $O(nr^2\log^4{n}\log{(\frac{1}{\epsilon})})$ & $O(mn^{2}\log{(\frac{1}{\epsilon})})$  \tabularnewline
\hline 
AP-ROM  \cite{soltani2018towards}& $O(nr^3\log^4{n}\log{(\frac{1}{\epsilon})})$ & $O(mnr\log{n}\log{(\frac{1}{\epsilon})})$  \tabularnewline
  \hline \hline \hline 
  Algorithms without resampling & Sample complexity & Computational complexity   \tabularnewline
\hline \hline
Convex \cite{chen2013exact} &  $O(nr)$ & $O(mn^{2}\frac{1}{\sqrt{\epsilon}})$ \tabularnewline \hline
GD \cite{sanghavi2017local}  & $O(nr^{6}\log^{2}{n})$ & $O(mn^{5}r^{3}\log^{4}{n}\log{(\frac{1}{\epsilon})})$   \tabularnewline
\hline 
GD (Algorithm~\ref{alg:ncvx_sketching}, Ours) & $O(nr^{4}\log{n})$ & $O(mnr \max\{\log^{2}{n}, r^{2}\} \log{(\frac{1}{\epsilon})})$   \tabularnewline  \hline
\end{tabular}
\end{center}
\caption{Comparisons with existing results in terms of sample complexity and computational complexity to reach $\epsilon$-accuracy. The top half of the  table is concerned with algorithms that require resampling, while the bottom half of the table covers algorithms without resampling. } \label{table:existing}
\end{table}


Several other existing works have suggested different approaches for low-rank matrix factorization from rank-one measurements, of which the statistical and computational guarantees to reach $\epsilon$-accuracy are summarized in Table~\ref{table:existing}. We note our guarantee is the only one that achieves simultaneous near-optimal sample complexity and computational complexity. Iterative algorithms based on alternating minimization or noisy power iterations \cite{zhong2015efficient,lin2016non,soltani2018towards} require a {\em fresh} set of samples at every iteration, which is never executed in practice, and the sample complexity grows unbounded for {\em exact} recovery.

Many nonconvex methods have been proposed and analyzed recently to solve the phase retrieval problem, including the Kaczmarz method \cite{wei2015solving,tan2017phase,jeong2017convergence} and approximate message passing \cite{ma2018optimization}. In \cite{chi2016kaczmarz}, the Kaczmarz method is generalized to solve the problem studied in this paper, but no theoretical performance guarantees are provided.

The local geometry studied in our paper is in contrast to \cite{sun2016geometric}, which studied the global landscape of phase retrieval, and showed that there are no spurious local minima as soon as the sample complexity is above $O(n\log^3 n)$. It will be interesting to study the landscape property of the generalized model in our paper.


Our model is  also related to learning shallow neural networks. \cite{zhong17a} studied the performance of gradient descent with resampling and an initialization provided by the tensor method for various activation functions, however their analysis did not cover quadratic activations. For quadratic activations, \cite{livni2014computational} adopts a greedy learning strategy, and can only guarantee sublinear convergence rate. Moreover,   \cite{soltanolkotabi2017theoretical} studied the optimization landscape for an over-parameterized shallow neural network with quadratic activation, where $r$ is larger than $n$.

\section{Outline of Theoretical Analysis}

This section provides the proof sketch of the main results, with the details  deferred to the appendix. Our theoretical analysis is inspired by the work of \cite{ma2017implicit} for phase retrieval and follows  the general recipe outlined in \cite{ma2017implicit}, while significant changes and elaborate derivations are needed. We refine the analysis to show that both the signal reconstruction error and the entry-wise error contract linearly, where the latter is not revealed by \cite{ma2017implicit}. In below, we first characterize a region of incoherence and contraction that enjoys both strong convexity and smoothness along certain directions. We then  demonstrate --- via an induction argument --- that the iterates always stay within this nice region. Finally, the proof is complete by validating the desired properties of spectral initialization.

\subsection{Local Geometry and Error Contraction}
We start with characterizing a local region around $\bX^{\natural}$, within which the loss function enjoys desired restricted strong convexity and smoothness properties. This requires exploring the property of the Hessian of $f(\bX)$, which is given by
\begin{equation}\label{eq:hessian}
\nabla^{2}f(\bX)
= \frac{1}{m}\sum_{i=1}^{m} \left[ \left(\big\Vert \ba_{i}^{\top}\bX\big\Vert_{2}^{2} - y_{i}\right)\bI_{r} + 2\bX^{\top}\ba_{i}\ba_{i}^{\top}\bX \right] \otimes \big( \ba_{i}\ba_{i}^{\top} \big).
\end{equation} 
Here, we use $\otimes$ to denote the Kronecker product and hence $\nabla^{2}f(\bX) \in \mathbb{R}^{nr \times nr}$. 
Now we are ready to state the following lemma regarding this local region, which will be referred to as the region of incoherence and contraction (RIC) throughout this paper. The proof is given in Appendix~\ref{proof_lemma_restrict_concen_hessian_neighbor}.

\begin{lemma}
\label{lemma_restrict_concen_hessian_neighbor}
Suppose the sample size obeys $m\ge c \frac{ \Vert \bX^{\natural} \Vert_{\F}^{4}}{\sigma_{r}^{4}\left(\bX^{\natural}\right) } n r \log{\left(n \kappa\right)}$ for some sufficiently large  constant $c>0$. Then with probability at least $1 - c_{1} n^{-12} - m e^{-1.5n} - m n^{-12} $, we have 
\begin{equation}
\label{eq:lower_bound}
\mathrm{vec}\left(\bV\right)^{\top} \nabla^{2}f(\bX) \mathrm{vec}\left(\bV\right) \ge 1.026  \sigma_{r}^{2} (\bX^{\natural} ) \left\Vert \bV \right\Vert_{\F}^{2},
\end{equation}
and 
\begin{equation}
\label{eq:upper_bound}
\left\Vert  \nabla^{2}f(\bX)\right\Vert \le 1.5 \sigma_{r}^{2} (\bX^{\natural} ) \log{n}    + 6\big\Vert\bX^{\natural}\big\Vert_{\F}^{2}
\end{equation}
hold simultaneously for all matrices $\bX$ and $\bV$ satisfying the following constraints: 
\begin{subequations}
\label{eq:def_RIC}
\begin{align}
\big\Vert\bX - \bX^{\natural}\big\Vert_{\F} & \le \frac{1}{24}\frac{\sigma_{r}^{2}\left(\bX^{\natural}\right)}{\left\Vert \bX^{\natural} \right\Vert_{\F}}, \label{eq:def_local} \\
\max_{1\le l\le m} \Big\Vert \ba_{l}^{\top}\big(\bX - \bX^{\natural}\big) \Big\Vert_{2} & \le \frac{1}{24} \sqrt{\log{n}}  \cdot \frac{\sigma_{r}^{2}\left(\bX^{\natural}\right)}{\Vert \bX^{\natural} \Vert_{\F}},
	\label{eq:def_incoherence} 
\end{align}
\end{subequations}
and $\bV = \bT_{1}\bQ_{\bT} - \bT_{2}$ satisfying 
\begin{equation}\label{eq:cond_direction}
\big\Vert \bT_{2} - \bX^{\natural}\big\Vert \le  \frac{1}{24}\frac{\sigma_{r}^{2}\left(\bX^{\natural}\right)}{\left\Vert \bX^{\natural} \right\Vert} , 
\end{equation}
where $\bQ_{\bT} := \argmin_{\bP\in\mathcal{O}^{r\times r} } \left\Vert \bT_{1}\bP - \bT_{2}\right\Vert_{\F}$. Here, $c_1$ is some absolute positive constant. 
\end{lemma}


The condition \eqref{eq:def_RIC} on $\bX$ formally characterizes the RIC, which enjoys the claimed restricted strong convexity (see (\ref{eq:lower_bound})) and smoothness (see (\ref{eq:upper_bound})). With Lemma~\ref{lemma_restrict_concen_hessian_neighbor} in mind, it is easy to see that if $\bX_t$ lies within the RIC, the estimation error shrinks in the presence of a properly chosen step size. This is given in the lemma below whose proof can be found in Appendix~\ref{proof_lemma_induction}.



\begin{lemma}\label{lemma:lemma_induction}
Suppose the sample size obeys $m\ge c \frac{ \Vert \bX^{\natural} \Vert_{\F}^{4}}{\sigma_{r}^{4}\left(\bX^{\natural}\right) } n r \log{\left(n \kappa\right)}$ for some sufficiently large constant $c>0$. Then with probability at least $1 - c_{1} n^{-12} - m e^{-1.5n} - m n^{-12} $, if $\bX_t$ falls within the RIC as described in  \eqref{eq:def_RIC}, we have 
\begin{equation*}
\mathrm{dist} \big(\bX_{t+1}, \bX^{\natural}\big) \leq \left( 1-  0.513\mu \sigma_r^2(\bX^{\natural})  \right)  \mathrm{dist}\big(\bX_{t}, \bX^{\natural}\big) ,
\end{equation*} 
provided that the step size obeys $0<\mu_{t} \equiv \mu \le \frac{1.026 \sigma_{r}^{2}\left(\bX^{\natural}\right)}{\big(1.5 \sigma_{r}^{2} (\bX^{\natural}) \log{n}    + 6\Vert\bX^{\natural} \Vert_{\F}^{2}\big)^{2}}$. Here, $c_1>0$ is some universal constant. 
\end{lemma}

Assuming that the iterates $\{\bX_{t}\}$, stay within the RIC (see \eqref{eq:def_RIC}) for the first $T_{c}$ iterations, according to Lemma~\ref{lemma:lemma_induction}, we have, by induction, that 
\begin{equation*}
\mathrm{dist} \big(\bX_{T_{c}+1}, \bX^{\natural}\big) \leq \left( 1-  0.513\mu \sigma_r^2(\bX^{\natural})  \right)^{T_{c}+1}  \mathrm{dist}\big(\bX_{0}, \bX^{\natural}\big) \leq   \frac{1}{24\sqrt{6}} \cdot \frac{\sqrt{\log{n}}}{\sqrt{n}} \cdot  \frac{\sigma_{r}^{2}\left(\bX^{\natural}\right)}{\left\Vert\bX^{\natural}\right\Vert_{\F}}  
\end{equation*} 
as soon as 
\begin{equation}\label{eq:iterations}
T_{c} \ge c \max{\left\{\log^{2}{n}, \frac{ \left\Vert\bX^{\natural}\right\Vert_{\F}^{4}} {  \sigma_{r}^{4}\left(\bX^{\natural}\right)} \right\}} \log n,
\end{equation}
for some large enough constant $c$. The iterates when $t\geq T_{c}$ are easier to deal with; in fact, it is easily seen that $\bX_{t+1}$ stays in the RIC since
\begin{align}
\max_{1\le l\le m} \left\Vert \ba_{l}^{\top} \big( \bX_{t+1}\bQ_{t+1} - \bX^{\natural} \big) \right\Vert_{2}
& \le \max_{1\le l\le m} \big\Vert \ba_{l} \big\Vert_{2} \left\Vert \bX_{t+1}\bQ_{t+1} - \bX^{\natural} \right\Vert \nonumber \\
& \le   \sqrt{6n} \cdot  \frac{1}{24\sqrt{6}} \cdot \frac{\sqrt{\log{n}}}{\sqrt{n}} \cdot  \frac{\sigma_{r}^{2}\left(\bX^{\natural}\right)}{\left\Vert\bX^{\natural}\right\Vert_{\F}} \label{equ_t_alog_asquare}\\
&  = \frac{1}{24} \sqrt{\log{n}} \cdot  \frac{\sigma_{r}^{2}\left(\bX^{\natural}\right)}{\left\Vert\bX^{\natural}\right\Vert_{\F}}, \nonumber
\end{align}
where \eqref{equ_t_alog_asquare} follows from Lemma~\ref{lemma_a_sqrtroot_concen} for all $t\geq T_{c}$. Consequently,  contraction of the estimation error $\mathrm{dist}\big(\bX_{t}, \bX^{\natural}\big) $ can be guaranteed by Lemma~\ref{lemma_restrict_concen_hessian_neighbor} for all $t\geq T_{c}$ with probability at least $1 - c_{1} n^{-12} - m e^{-1.5n} - m n^{-12} $.


\subsection{Introducing Leave-One-Out Sequences}

It has now become clear that the key remaining step is to  verify that the iterates $\{\bX_{t}\}$ satisfy \eqref{eq:def_RIC} for the first $T_{c}$ iterations, where $T_{c}$ is on the order of \eqref{eq:iterations}. Verifying (\ref{eq:def_incoherence}) is conceptually hard since the iterates $\{\bX_t\}$ are statistically dependent with all the sensing vectors $\{\ba_{i}\}_{i=1}^{m}$. To tackle this problem, for each $1 \leq l \leq m$, we introduce an auxiliary leave-one-out sequence $\{\bX^{(l)}_{t}\}$, which discards a single measurement from consideration. Specifically, the sequence $\{\bX^{(l)}_{t}\}$ is the gradient iterates operating on the following leave-one-out function
\begin{equation}
f^{(l)}\left(\bX\right) := \frac{1}{4m} \sum_{i: i\neq l} \left(y_{i} - \big\Vert \ba_{i}^{\top}\bX\big\Vert_{2}^{2}\right)^{2}.
\end{equation} 
See Algorithm~\ref{alg:loo} for a formal definition of the leave-one-out sequences. Again, we want to emphasize that Algorithm~\ref{alg:loo} is just an auxiliary procedure useful for the theoretical analysis, and it does not need to be implemented in practice.

\begin{algorithm}[t]
\caption{Leave-One-Out Versions}
\label{alg:loo}
\noindent \textbf{Input:} Measurements $\left\{y_{i}\right\}_{i:i\neq l}$, and sensing vectors $\left\{\ba_{i}\right\}_{i:i\neq l}$.

\noindent \textbf{Parameters:} Step size $\mu_{t}$, rank $r$, and number of iterations $T$.

\noindent \textbf{Initialization:} $\bX_{0}^{(l)} = \bZ_{0}^{(l)}\boldsymbol{\Lambda}_{0}^{(l) 1/2}$, where the columns of $\bZ_{0}^{(l)} \in\mathbb{R}^{n\times r}$ contain the normalized eigenvectors corresponding to the $r$ largest eigenvalues of the matrix 
\begin{equation}
\bY^{(l)} = \frac{1}{2m} \sum_{i: i\neq l} y_{i}\ba_{i}\ba_{i}^{\top},
\end{equation}
and $\boldsymbol{\Lambda}_{0}^{(l)}$ is an $r\times r$ diagonal matrix, with the entries on the diagonal given as 
\begin{equation}
\left[\boldsymbol{\Lambda}_{0}^{(l)}\right]_{i} = \lambda_{i}\big(\bY^{(l)}\big) - \lambda^{(l)}, \quad   \ i=1,\cdots,r,
\end{equation}
where $\lambda^{(l)} = \frac{1}{2m} \sum_{i: i\neq l} y_{i}$ and $\lambda_{i}\left(\bY^{(l)}\right)$ is the $i$th largest eigenvalue of $\bY^{(l)}$.

\noindent \textbf{Gradient loop:} For $t = 0:1:T-1$, do 
\begin{equation}
\bX_{t+1}^{(l)} = \bX_{t}^{(l)} - \mu_{t} \cdot \frac{1}{m} \sum_{i: i\neq l } \left( \big\Vert \ba_{i}^{\top}\bX_{t}^{(l)} \big\Vert_{2}^{2} - y_{i} \right) \ba_{i}\ba_{i}^{\top}\bX_{t}^{(l)}.
\end{equation}

\noindent \textbf{Output:} $\bX_{T}^{(l)}$.

\end{algorithm}


\subsection{Establishing Incoherence via Induction}\label{sec_inductive_update}

Our proof is inductive in nature with the following induction hypotheses:
\begin{subequations}\label{eq:induction_hyp}
\begin{align}
\left\Vert \bX_{t}\bQ_{t} - \bX^{\natural} \right\Vert_{\F} 
& \le C_{1} \left(1 - 0.5 \sigma_{r}^{2}\left(\bX^{\natural}\right) \mu \right)^{t} \frac{\sigma_{r}^{2}\left(\bX^{\natural}\right)}{\left\Vert \bX^{\natural} \right\Vert_{\F}}, \label{equ_inductive_att_c1}\\
 \max_{1\le l\le m} \left\Vert \bX_{t}\bQ_{t} - \bX_{t}^{(l)}\bR_{t}^{(l)} \right\Vert_{\F} 
 & \le C_{3} \left(1 - 0.5 \sigma_{r}^{2}\left(\bX^{\natural}\right) \mu \right)^{t} \sqrt{\frac{\log{n}}{n}} \cdot  \frac{\sigma_{r}^{2}\left(\bX^{\natural}\right)}{\kappa \left\Vert \bX^{\natural} \right\Vert_{\F}}, \label{equ_inductive_att_c3}\\
  \max_{1\le l\le m} \left\Vert \ba_{l}^{\top}\left(\bX_{t}\bQ_{t}-\bX^{\natural}\right) \right\Vert_{2} 
 & \le C_{2} \left(1 - 0.5 \sigma_{r}^{2}\left(\bX^{\natural}\right) \mu \right)^{t} \sqrt{\log{n}}  \cdot \frac{\sigma_{r}^{2}\left(\bX^{\natural}\right)}{\left\Vert \bX^{\natural} \right\Vert_{\F}}, \label{equ_inductive_att_c2}
\end{align}
\end{subequations}
where $\bR_{t}^{(l)}  = \argmin_{\bP\in\mathcal{O}^{r\times r}} \big\Vert \bX_{t}\bQ_{t} - \bX_{t}^{(l)}\bP \big\Vert_{\F} $, and the positive constants $C_1$, $C_2$ and $C_3$ satisfy
\begin{equation}\label{eq:constant_cond}
C_{1} + C_{3} \le \frac{1}{24}, \quad C_{2} + \sqrt{6}C_{3} \le \frac{1}{24},\quad   5.86 C_{1} +  29.3 C_{3} + 5\sqrt{6}C_{3}  \le  C_{2}.
\end{equation}
Furthermore, the step size $\mu$ is chosen as 
\begin{equation}\label{equ_induc_stepsize}
\mu = \frac{c_{0}\sigma_{r}^{2}\left(\bX^{\natural}\right)}{\big( \sigma_{r}^{2} (\bX^{\natural}) \log{n}    +  \Vert\bX^{\natural} \Vert_{\F}^{2}\big)^{2}}
\end{equation}
with appropriate universal constant $c_{0}>0$.

Our goal is to show that if the $t$th iteration $\bX_t$ satisfies the induction hypotheses \eqref{eq:induction_hyp}, then the $(t+1)$th iteration $\bX_{t+1}$ also satisfies \eqref{eq:induction_hyp}. It is straightforward to see that the hypothesis \eqref{equ_inductive_att_c1} has already been established by Lemma~\ref{lemma:lemma_induction}, and we are left with \eqref{equ_inductive_att_c3} and \eqref{equ_inductive_att_c2}. We first establish \eqref{equ_inductive_att_c3} in the following lemma, which measures the proximity between $\bX_t$ and the leave-one-out versions $\bX_t^{(l)}$, whose proof is provided in Appendix~\ref{proof:lemma_proximity}. 


\begin{lemma}
\label{lemma:proximity}
Suppose the sample size obeys $m\ge c \frac{ \Vert \bX^{\natural} \Vert_{\F}^{4}}{\sigma_{r}^{4}\left(\bX^{\natural}\right) } n r \log{\left(n \kappa\right)}$ for some sufficiently large constant $c>0$. If the induction hypotheses \eqref{eq:induction_hyp} hold for the $t$th iteration, with probability at least $1 - c_{1} n^{-12} - me^{-1.5n} - mn^{-12} $, we have 
\begin{equation*}
\max_{1\le l\le m}\left\Vert \bX_{t+1}\bQ_{t+1} - \bX_{t+1}^{(l)}\bR_{t+1}^{(l)}\right\Vert_{\F}
 \le C_{3} \left(1 - 0.5 \sigma_{r}^{2}\left(\bX^{\natural}\right) \mu \right)^{t+1} \sqrt{\frac{\log{n}}{n}} \cdot \frac{\sigma_{r}^{2}\left(\bX^{\natural}\right)}{\kappa \left\Vert\bX^{\natural} \right\Vert_{\F}},
\end{equation*}
as long as the step size obeys \eqref{equ_induc_stepsize}. Here, $c_1>0$ is some absolute constant.
\end{lemma}

In addition, the  incoherence property of $\bX_{t+1}^{(l)}$ with respect to the $l$th sensing vector $\ba_l$ is  relatively easier to establish, due to their statistical independence. Combined with the proximity bound from Lemma~\ref{lemma:proximity}, this allows us to justify the incoherence property of the original iterates $\bX_{t+1}$, as summarized in the lemma below, whose proof is given in Appendix~\ref{proof_incoherence_induction}. 
\begin{lemma}\label{lemma:incoherence_induction}
Suppose the sample size obeys $m\ge c \frac{ \Vert \bX^{\natural} \Vert_{\F}^{4}}{\sigma_{r}^{4}\left(\bX^{\natural}\right) } n r \log{\left(n \kappa\right)}$ for some sufficiently large constant $c>0$. If the induction hypotheses \eqref{eq:induction_hyp} hold for the $t$th iteration, with probability exceeding $1 - c_{1} n^{-12}  - me^{-1.5n} - 2mn^{-12}$,
\begin{equation*}
\max_{1\le l\le m} \left\Vert \ba_{l}^{\top} \left( \bX_{t+1}\bQ_{t+1} - \bX^{\natural} \right) \right\Vert_{2} 
\le C_{2}  \left(1 - 0.5 \sigma_{r}^{2}\left(\bX^{\natural}\right) \mu \right)^{t+1} \sqrt{\log{n}} \cdot \frac{\sigma_{r}^{2}\left(\bX^{\natural}\right)}{\left\Vert \bX^{\natural} \right\Vert_{\F}}
\end{equation*}
holds as long as the step size satisfies \eqref{equ_induc_stepsize}. Here, $c_1>0$ is some universal constant.
\end{lemma}


\subsection{Spectral Initialization}
Finally, it remains to verify that the induction hypotheses hold for the initialization, i.e.~the base case when $t=0$. This is supplied by the following lemma, whose proof is given in Appendix~\ref{proof_lemma_initialization}.
\begin{lemma}\label{lemma:initialization}
Suppose that the sample size exceeds $m\ge c \max\left\{ \frac{\left\Vert \bX^{\natural} \right\Vert_{\F}}{\sigma_{r}\left( \bX^{\natural} \right) } \sqrt{r},  \kappa  \right\} \frac{\left\Vert \bX^{\natural} \right\Vert_{\F}^{5}}{\sigma_{r}^{5}\left( \bX^{\natural} \right) }  n \sqrt{r} \log{n}$ for some sufficiently large constant $c>0$. Then $\bX_0$ satisfies \eqref{eq:induction_hyp} with probability at least $1 - c_{1}n^{-12} -  me^{-1.5n} - 3mn^{-12}$, where $c_1$ is some absolute positive constant. 
\end{lemma}




\section{Conclusions}

In this paper, we have shown that low-rank positive semidefinite matrices can be recovered from a near-minimal number of random rank-one measurements, via the vanilla gradient descent algorithm following spectral initialization. Our results significantly improve upon existing results in several ways, both computationally and statistically. In particular, our algorithm does not require resampling at every iteration (and hence requires fewer samples). 
The gradient iteration can provably employ a much more aggressive step size than what was suggested in prior literature (e.g.~\cite{sanghavi2017local}), thus resulting in much smaller iteration complexity and hence lower computational cost. 
All of this is enabled by establishing the implicit regularization feature of gradient descent for nonconvex statistical estimation, where the iterates remain incoherent with the sensing vectors throughout the execution of the whole algorithm.

There are several problems that are worth exploring in future investigation.  For example, our theory reveals the typical size of the  fitting  error of $\bX_t$ (i.e.~$y_i - \| \ba_i^{\top} \bX^{\natural} \|_2$) in the presence of noiseless data,  which would serve as a helpful benchmark when 
  separating sparse outliers in the more realistic scenario.   Another direction is to  explore whether implicit regularization remains valid for learning shallow neural networks \cite{zhong17a}. Since the current work can be viewed as learning a one-hidden-layer fully-connected network with a quadratic activation function $\sigma(z)=z^2$, it would be of great interest to study if the techniques utilized herein can be used to develop strong guarantees when the activation function takes other forms.

\section*{Acknowledgements}
The work of Y. Li and Y. Chi is supported in part by AFOSR under the grant FA9550-15-1-0205, by ONR under the grant N00014-18-1-2142, and by NSF under the grants CAREER ECCS-1818571 and CCF-1704245. 



\appendix
\section*{Appendices}

\section{Technical Lemmas}
In this section, we document a few useful lemmas that are used throughout the proof.
\begin{lemma}\cite[Lemma 5.4]{tu2016low}\label{lemma_fronorm_lowbound}
For any matrices $\bX$, $\bU \in \mathbb{R}^{n \times r}$, we have 
\begin{equation*}
\big\Vert \bX\bX^{\top} - \bU\bU^{\top} \big\Vert_{\F} \ge \sqrt{2(\sqrt{2}-1)} \sigma_{r} \left( \bX \right) \mathrm{dist}(\bX,\bU).
\end{equation*}
\end{lemma}

\begin{lemma}[Covering number for low-rank matrices]\cite[Lemma 3.1]{candes2011tight}\label{lemma_covering_net} 
	Let $\mathcal{S}_r = \{\bX\in\mathbb{R}^{n_1\times n_2},\mathrm{rank}(\boldsymbol{X})\leq r, \| \boldsymbol{X}\|_{\F}=1\}$. Then there exists an $\epsilon$-net $\bar{\mathcal{S}}_r\subset \mathcal{S}_r$ with respect to the Frobenius norm obeying $ \left\vert \bar{\mathcal{S}}_{r} \right\vert \leq (9/\epsilon)^{(n_1+n_2+1)r}$. 
\end{lemma}

\begin{lemma}\cite{bentkus2003inequality, candes2015phase}
\label{lemma_oneside_tail_bound}
Suppose $x_{1}, \cdots, x_{m}$ are i.i.d.~real-valued random variables obeying $x_{i} \le  b$ for some deterministic number $b >0$, $\mathbb{E}\left[x_{i}\right] = 0$, and $\mathbb{E}\left[x_{i}^{2}\right] = d^{2}$. Setting $\sigma^{2} = m \cdot \max\{b^{2}, d^{2}\}$, we have 
\begin{equation*}
\mathbb{P}\left(\sum_{i=1}^{m} x_{i} \ge t\right) \le \min\left\{\exp{\left(-\frac{t^{2}}{2\sigma^{2}}\right)}, 25\left(1 - \Phi\left(\frac{t}{\sigma}\right) \right)\right\},
\end{equation*}
where $\Phi(\cdot)$ is the cumulative distribution function of a standard Gaussian variable.
\end{lemma}

\begin{lemma}\cite[Theorem 5.39]{Vershynin2012}\label{lemma_aaT_operator_concentration}
Suppose the $\ba_{i}$'s are i.i.d.~random vectors following $\ba_{i}\sim \mathcal{N}\left(\boldsymbol{0}, \bI_n\right)$,  $i=1,\cdots,m$. Then for every $t\ge0$ and $0< \delta\le 1$, 
\begin{equation*}
\left\Vert \bI_{n} - \frac{1}{m} \sum_{i=1}^{m} \ba_{i}\ba_{i}^{\top} \right\Vert \le \delta
\end{equation*}
holds with probability at least $1 - 2e^{-ct^{2}}$, where $\delta = C\sqrt{\frac{n}{m}} + \frac{t}{\sqrt{m}}$. On this event,  for all $\bW\in\mathbb{R}^{n\times r}$, there exists
\begin{equation*}
 \left\vert  \frac{1}{m} \sum_{i=1}^{m} \big\Vert \ba_{i}^{\top}\bW\big\Vert_{2}^{2} - \left\Vert \bW \right\Vert_{\F}^{2}  \right\vert \le  \delta \left\Vert \bW \right\Vert_{\F}^{2}.
\end{equation*} 
\end{lemma}

\begin{lemma}\cite{candes2015phase}\label{lemma_a_sqrtroot_concen}
Suppose the $\ba_{i}$'s are i.i.d.~random vectors following $\ba_{i}\sim \mathcal{N}\left(\boldsymbol{0}, \bI_n\right)$, $i=1,\cdots,m$. Then with probability at least $1-me^{-1.5n}$, we have 
\begin{equation*}
\max_{1\le i\le m} \left\Vert \ba_{i}\right\Vert_{2} \le \sqrt{6n}.
\end{equation*}

\end{lemma}

\begin{lemma}\label{lemma_a_log_tight_concen}
Fix $\bW \in\mathbb{R}^{n\times r}$. Suppose the $\ba_{i}$'s are i.i.d.~random vectors following $\ba_{i}\sim \mathcal{N}\left(\boldsymbol{0}, \bI_n\right)$,  $i=1,\cdots,m$. Then with probability at least $1  - m r n^{-13}$, we have
\begin{equation*}
\max_{1\le i\le m} \big\Vert \ba_{i}^{\top} \bW \big\Vert_{2} \le 5.86 \sqrt{\log{n}} \left\Vert \bW \right\Vert_{\F}.
\end{equation*}
\end{lemma}
\begin{proof}
Define $\bW = [\bw_{1}, \bw_{2}, \cdots, \bw_{r} ]$, then we can write $\left\Vert \ba_{i}^{\top} \bW \right\Vert_{2}^{2} = \sum_{k=1}^{r} \left( \ba_{i}^{\top} \bw_{k} \right)^{2}$. Recognize that $\left( \ba_{i}^{\top} \frac{\bw_{k}}{\left\Vert\bw_{k}\right\Vert_{2}} \right)^{2}$ follows the $\chi^2$ distribution with $1$ degree of freedom.  It then follows from \cite[Lemma 1]{laurent2000adaptive} that  
\begin{equation*}
\mathbb{P}\left( \left( \ba_{i}^{\top} \frac{\bw_{k}}{\left\Vert\bw_{k}\right\Vert_{2}} \right)^{2} \ge 1 + 2\sqrt{t} + 2t \right) \le \exp{\left(-t\right)},
\end{equation*}
for any $t > 0$. Taking $t = 13 \log{n}$ yields 
\begin{equation*}
\mathbb{P}\left( \left( \ba_{i}^{\top} \bw_{k} \right)^{2} \le 34.3 \left\Vert\bw_{k}\right\Vert_{2}^{2} \log{n}  \right) \ge 1-  n^{-13}.
\end{equation*}
Finally, taking the union bound, we obtain 
\begin{equation*}
\max_{1\le i\le m} \big\Vert \ba_{i}^{\top} \bW \big\Vert_{2}^{2} \le \sum_{k=1}^{r} 34.3 \left\Vert\bw_{k}\right\Vert_{2}^{2} \log{n}  = 34.3 \left\Vert \bW\right\Vert_{\F}^{2} \log{n} 
\end{equation*}
with probability at least $1 - m r n^{-13}$.
\end{proof}


\begin{lemma}\label{lemma_expectation_gaussian}
Suppose  $\ba\sim \mathcal{N}\left(\boldsymbol{0}, \bI_{n}\right)$. Then for any fixed matrices $\bX$, $\bH\in\mathbb{R}^{n\times r}$, we have
\begin{align*}
\mathbb{E}\left[ \big\Vert \ba^{\top}\bH\big\Vert_{2}^{2}  \big\Vert \ba^{\top}\bX \big\Vert_{2}^{2}\right]
&  = \big\Vert\bH\big\Vert_{\F}^{2} \big\Vert\bX\big\Vert_{\F}^{2} + 2 \big\Vert\bH^{\top}\bX\big\Vert_{\F}^{2};\\
\mathbb{E}\left[\big(\ba^{\top}\bH\bX^{\top}\ba\big)^{2}\right]
& =  \left( \mathrm{Tr}\big(\bH^{\top}\bX\big)\right)^{2} + \mathrm{Tr}\big(\bH^{\top}\bX\bH^{\top}\bX\big) + \big\Vert \bH\bX^{\top}\big\Vert_{\F}^{2}.
\end{align*}
Moreover, for any order $k\geq 1$, we have $\mathbb{E}\big[ \Vert \ba^{\top}\bH \Vert_{2}^{2k}\big] 
 \le c_k \left\Vert \bH \right\Vert_{\F}^{2k}$,
where $c_k>0$ is a numerical constant that  depends only on $k$.
\end{lemma}

\begin{proof}
Let $\bX =[ \bx_{1}, \bx_{2}, \cdots, \bx_{r}] $ and $\bH = [\bh_{1}, \bh_{2}, \cdots, \bh_{r} ]$. Based on the simple facts
\begin{align*}
\mathbb{E}\left[(\bx^{\top}\ba)^{2}\ba\ba^{\top}\right] & = \left\Vert\bx\right\Vert_{2}^{2}\bI_{n} + 2\bx\bx^{\top},\\
\mathbb{E}\left[(\ba^{\top}\bx_{i})(\ba^{\top}\bx_{j})\ba\ba^{\top}\right] & = \bx_{i}\bx_{j}^{\top} + \bx_{j}\bx_{i}^{\top} + \bx_{i}^{\top}\bx_{j}\bI_{n},
\end{align*}
we can derive 
\begin{align*}
\mathbb{E}\left[ \left\Vert \ba^{\top}\bH\right\Vert_{2}^{2}  \left\Vert \ba^{\top}\bX \right\Vert_{2}^{2}\right]
& =  \sum_{i=1}^{r}\sum_{j=1}^{r}\mathbb{E}\left[ \left(\ba^{\top}\bh_{i}\right)^{2}  \left(\ba^{\top}\bx_{j}\right)^{2}\right] \\
&= \sum_{i=1}^{r}\sum_{j=1}^{r} \left[ \left\Vert \bh_{i}\right\Vert_{2}^{2} \left\Vert \bx_{j}\right\Vert_{2}^{2} + 2\left(\bh_{i}^{\top}\bx_{j}\right)^{2} \right] \\
& = \left\Vert\bH\right\Vert_{\F}^{2} \left\Vert\bX\right\Vert_{\F}^{2} + 2 \big\Vert\bH^{\top}\bX\big\Vert_{\F}^{2},
\end{align*}
and
\begin{align*}
  \mathbb{E}\left[\left(\ba^{\top}\bH\bX^{\top}\ba\right)^{2}\right]
& = \mathbb{E}\left[ \sum_{i=1}^{r}\left(\ba^{\top}\bh_{i}\right)^{2}\left(\ba^{\top}\bx_{i}\right)^{2} + \sum_{ i \neq j} \left(\ba^{\top}\bh_{i}\right) \left(\ba^{\top}\bx_{i}\right) \left(\ba^{\top}\bh_{j}\right) \left(\ba^{\top}\bx_{j}\right) \right] \\
& = \sum_{i=1}^{r} \left[ \left\Vert \bh_{i}\right\Vert_{2}^{2} \left\Vert \bx_{i}\right\Vert_{2}^{2} + 2\left(\bh_{i}^{\top}\bx_{i}\right)^{2}  \right]  \\
& \quad + \sum_{ i \neq j} \left[ \left(\bh_{i}^{\top}\bx_{i}\right)  \left(\bh_{j}^{\top}\bx_{j}\right) +  \left(\bh_{i}^{\top}\bh_{j}\right) \left(\bx_{i}^{\top}\bx_{j}\right) +  \left(\bh_{i}^{\top}\bx_{j}\right)\left(\bx_{i}^{\top}\bh_{j}\right) \right] \\
& =  \left( \mathrm{Tr}\big(\bH^{\top}\bX\big)\right)^{2}  + \big\Vert \bH\bX^{\top}\big\Vert_{\F}^{2} + \mathrm{Tr}\big(\bH^{\top}\bX\bH^{\top}\bX\big).
\end{align*}

Finally, to bound $\mathbb{E}\left[ \left\Vert \ba^{\top}\bH\right\Vert_{2}^{2k}\right]$ for an arbitrary $\bH\in\mathbb{R}^{n\times r}$, we write the singular value decomposition of $\bH$ as $\bH= \bU\boldsymbol{\Sigma}\bV^{\top}$, where $\bU =[ \bu_{1}, \bu_{2}, \cdots, \bu_{r} ]\in\mathbb{R}^{n\times r}$, $\boldsymbol{\Sigma} = \mathrm{diag}\left\{\sigma_{1}, \sigma_{2}, \cdots, \sigma_{r}\right\}$, and $\bV\in\mathbb{R}^{r\times r}$.
This gives 
	$$\big\Vert \ba^{\top}\bH\big\Vert_{2}^{2} =\sum_{i=1}^r \sigma_{i}^2(\ba^{\top}\bu_{i})^2.$$ 
Let $b_{i} = \sigma_{i}\ba^{\top}\bu_{i}$  for $i=1,\cdots,r$, which are independent random variables obeying
	 $b_i \sim \mathcal{N}\left(0, \sigma_{i}^{2}\right)$ due to the fact $\bU^{\top}\bU = \bI_{r}$. Since $\mathbb{E}\left[b_{i}^{2t}\right] = \sigma_{i}^{2t} \left(2t-1\right)!! \le c_k \sigma_{i}^{2t}$ for any $i=1,\cdots,r$ and $t=1,\cdots,k$, where $c_k$ is some large enough constant depending only on $k$, we arrive at 
	$$\mathbb{E}\left[ \left(\sum_{i=1}^{r} b_{i}^{2} \right)^{k}\right] \le c_k \left(\sum_{i=1}^{r} \sigma_{i}^{2} \right)^{k} = c_{k} \left\Vert \bH \right\Vert_{\F}^{2k}$$
	as claimed. 
\end{proof}

 \begin{lemma}
\label{lemma_intial_weighted_mat_concen}
Fix $\bX^{\natural} \in \mathbb{R}^{n \times r}$. Suppose the $\ba_{i}$'s are i.i.d.~random vectors following $\ba_{i}\sim \mathcal{N}\left(\boldsymbol{0}, \bI_n\right)$, $i=1,\cdots,m$. For any $0< \delta \le 1$, suppose $m \ge c\delta^{-2} n \log{n}$ for some sufficiently large  constant $c>0$. Then we have
\begin{equation*}
\left\Vert \frac{1}{m} \sum_{i=1}^{m} \big\|\ba_i^{\top}\bX^{\natural}\big\|_2^2\ba_{i}\ba_{i}^{\top} -  \big\Vert \bX^{\natural}\big\Vert_{\F}^{2}\,\bI_{n} - 2\bX^{\natural}\bX^{\natural \top} \right\Vert \le \delta \big\Vert \bX^{\natural}\big\Vert_{\F}^{2},
\end{equation*}
with probability at least $1 - c_{1} r n^{-13}$, where $c_1>0$ is some absolute constant. 
\end{lemma}
\begin{proof} 
This proof adapts the results of \cite[Lemma~7.4]{candes2015phase} with refining the probabilities. Let $\ba(1)$ be the first element of a vector $\ba\sim \mathcal{N}\left(\boldsymbol{0}, \bI_n\right)$. Based on \cite[Theorem 1.9]{schudy2011concentration}, we have
\begin{align*}
\mathbb{P}\left( \left\vert  \frac{1}{m} \sum_{i=1}^{m} \left(\ba_{i}(1)\right)^{2}  - 1 \right\vert  \ge \delta \right) 
& \le   e^{2} \cdot e^{ - \left(c_{1} \delta^{2} m \right)^{1/2}};\\
\mathbb{P}\left( \left\vert  \frac{1}{m} \sum_{i=1}^{m} \left(\ba_{i}(1)\right)^{4}  - 3 \right\vert  \ge \delta \right) 
& \le   e^{2} \cdot e^{ - \left(c_{2} \delta^{2} m \right)^{1/4}};\\
\mathbb{P}\left( \left\vert  \frac{1}{m} \sum_{i=1}^{m} \left(\ba_{i}(1)\right)^{6}  - 15 \right\vert  \ge \delta \right) 
& \le   e^{2} \cdot e^{ - \left(c_{3} \delta^{2} m \right)^{1/6}}.
\end{align*}
So, by setting $m \gg \delta^{-2} n $, we have 
\begin{equation}\label{equ_groundtruth_weight_mat_concen_moment}
\left\vert  \frac{1}{m} \sum_{i=1}^{m} \left(\ba_{i}(1)\right)^{2}  - 1 \right\vert  \le \delta, \ \left\vert  \frac{1}{m} \sum_{i=1}^{m} \left(\ba_{i}(1)\right)^{4}  - 3 \right\vert  \le \delta, \ \mathrm{and} \ \left\vert  \frac{1}{m} \sum_{i=1}^{m} \left(\ba_{i}(1)\right)^{6}  - 15 \right\vert  \le \delta, 
\end{equation}
with probability at least $1 - c_{4} n^{-13}$ for some constant $c_4 >0$. Moreover, following \cite[Lemma 1]{laurent2000adaptive}, we know 
\begin{equation*}
\mathbb{P}\left( \left(\ba_{i}(1)\right)^{2}  \ge 1 + 2\sqrt{t} + 2t \right) \le \exp{\left(-t\right)},
\end{equation*}
which gives 
\begin{equation*}
\mathbb{P} \left( \left(\ba_{i}(1)\right)^{2}  \ge 36.5   \log{m} \right)  \le  \exp{\left(- 14 \log{m} \right)} = m^{-14},
\end{equation*}
if setting $t = 14 \log{m}$. Therefore, as long as $m \ge c n$, we have 
\begin{equation}\label{equ_groundtruth_weight_mat_concen_max}
\max_{1\le i\le m}  \left\vert \ba_{i}(1)\right\vert \le \sqrt{36.5   \log{m}},
\end{equation}
with probability at least $1 - c_{5} n^{-13}$ for some constant $c_5 > 0$.

With \eqref{equ_groundtruth_weight_mat_concen_moment} and \eqref{equ_groundtruth_weight_mat_concen_max}, the results in \cite[Lemma~7.4]{candes2015phase} imply that for any $0 <\delta \le 1$, as soon as $m \ge c\delta^{-2} n \log{n}$ for some sufficiently large constant $c$, with probability at least $1-c_1 n^{-13}$,
\begin{equation*}
\left\Vert \frac{1}{m} \sum_{i=1}^{m} \big(\ba_{i}^{\top}\bx \big)^{2} \ba_{i}\ba_{i}^{\top} -  \left\Vert \bx  \right\Vert_{2}^{2}\bI - 2 \bx \bx^{\top}  \right\Vert \le \delta \left\Vert \bx  \right\Vert_{2}^{2}
\end{equation*}
holds for any fixed vector $\bx\in\mathbb{R}^{n}$. Let $\bX^{\natural} =[ \bx_{1}^{\natural}, \bx_{2}^{\natural}, \cdots, \bx_{r}^{\natural} ]$. Instantiating the above bound for the set of vectors $\bx_k^{\natural}$, $k=1,\ldots, r$ and taking the union bound, we have
\begin{align*}
\left\Vert \frac{1}{m} \sum_{i=1}^{m} \big\|\ba_i^{\top}\bX^{\natural}\big\|_2^2 \,\ba_{i}\ba_{i}^{\top} -  \big\Vert \bX^{\natural}\big\Vert_{\F}^{2}\bI - 2\bX^{\natural}\bX^{\natural \top} \right\Vert 
&  \le \sum_{k=1}^r \left\Vert \frac{1}{m} \sum_{i=1}^{m} \big(\ba_i^{\top}\bx_k^{\natural} \big)^2\ba_{i}\ba_{i}^{\top} -  \big\Vert \bx_k^{\natural} \big\Vert_{2}^{2}\,\bI - 2\bx_k^{\natural}\bx_k^{\natural \top} \right\Vert  
	 \\
& \quad \leq \delta \sum_{k=1}^r \big\Vert \bx_k^{\natural}\big\Vert_{2}^{2} = \delta \big\Vert \bX^{\natural}\big\Vert_{\F}^{2}.
\end{align*}
\end{proof}

\section{Proof of Lemma~\ref{lemma_restrict_concen_hessian_neighbor}}\label{proof_lemma_restrict_concen_hessian_neighbor}

The crucial ingredient for proving the lower bound \eqref{eq:lower_bound} is the following lemma, whose proof is provided in Appendix~\ref{proof_lemma_restrict_concen_hessian_lower_looseconst}.


\begin{lemma}\label{lemma_restrict_concen_hessian_lower_looseconst}
Suppose $m\ge c \frac{ \big\Vert \bX^{\natural} \big\Vert_{\F}^{4}}{\sigma_{r}^{4}\left(\bX^{\natural}\right) } n r \log{\left(n \kappa\right)}$ with some large enough positive constant $c$, then with probability at least $1 - c_{1} n^{-12} - m e^{-1.5n}$, we have 
\begin{equation}
\mathrm{vec}\left(\bV\right)^{\top} \nabla^{2}f(\bX) \mathrm{vec}\left(\bV\right) \ge  2 \mathrm{Tr}\left(\bX^{\natural \top}\bV\bX^{\natural \top}\bV\right)  + 1.204  \sigma_{r}^{2} (\bX^{\natural}  ) \left\Vert \bV \right\Vert_{\F}^{2},
\end{equation}
for all matrices $\bX$ and $\bV$ where $\bX $ satisfies $\left\Vert\bX - \bX^{\natural}\right\Vert_{\F}\le \frac{1}{24}\frac{\sigma_{r}^{2}\left(\bX^{\natural}\right)}{\left\Vert \bX^{\natural} \right\Vert_{\F}}$. Here, $c_{1} > 0$ is some universal constant.
\end{lemma}

With Lemma \ref{lemma_restrict_concen_hessian_lower_looseconst} in place, we are ready to prove \eqref{eq:lower_bound}. 
Let $\bV = \bT_{1}\bQ_{\bT} - \bT_{2}$ satisfy the assumptions in  Lemma~\ref{lemma_restrict_concen_hessian_neighbor}, then we can demonstrate that
\begin{align}
& \mathrm{Tr}\left(\bX^{\natural \top}\bV\bX^{\natural \top}\bV\right) \nonumber\\
& = \mathrm{Tr}\left( \big(\bX^{\natural}-\bT_{2}+\bT_{2}\big)^{\top}\bV \big(\bX^{\natural}-\bT_{2}+\bT_{2}\big)^{\top}\bV\right) \nonumber\\
& =  \mathrm{Tr}\left(\big(\bX^{\natural} - \bT_{2}\big)^{\top}\bV\big(\bX^{\natural} - \bT_{2}\big)^{\top}\bV\right) + 2\mathrm{Tr}\left(\big(\bX^{\natural}-\bT_{2}\big)^{\top}\bV\bT_{2}^{\top}\bV\right) + \mathrm{Tr}\left(\bT_{2}^{\top}\bV\bT_{2}^{\top}\bV\right) \nonumber\\
& \ge \mathrm{Tr}\left(\bT_{2}^{\top}\bV\bT_{2}^{\top}\bV\right) - \big\Vert \bX^{\natural} - \bT_{2} \big\Vert^{2} \left\Vert\bV\right\Vert_{\F}^{2} - 2\big\Vert \bX^{\natural} - \bT_{2} \big\Vert\left\Vert \bT_{2} \right\Vert \left\Vert\bV\right\Vert_{\F}^{2} \nonumber\\
& = \big\Vert \bT_{2}^{\top}\bV \big\Vert_{\F}^{2}  - \big\Vert \bX^{\natural}-\bT_{2} \big\Vert^{2} \left\Vert\bV\right\Vert_{\F}^{2} - 2\big\Vert \bX^{\natural}-\bT_{2} \big\Vert\left\Vert \bT_{2} \right\Vert \left\Vert\bV\right\Vert_{\F}^{2}  \label{equ_hessian_lower_loose_symm}\\
& \ge - \left[ \left( \frac{1}{24}\frac{\sigma_{r}^{2}\left(\bX^{\natural}\right)}{\left\Vert \bX^{\natural} \right\Vert}\right)^{2} + 2\cdot  \frac{1}{24}\frac{\sigma_{r}^{2}\left(\bX^{\natural}\right)}{\left\Vert \bX^{\natural} \right\Vert} \cdot \left( \frac{1}{24}\frac{\sigma_{r}^{2}\left(\bX^{\natural}\right)}{\big\Vert \bX^{\natural} \big\Vert} + \big\Vert\bX^{\natural}\big\Vert\right) \right] \left\Vert\bV\right\Vert_{\F}^{2} \label{equ_hessian_lower_loose_invoking_lemma2}\\
& \ge - 0.0886 \sigma_{r}^{2} (\bX^{\natural} ) \left\Vert \bV \right\Vert_{\F}^{2},  \label{equ_hessian_lower_loose_finalbound}
\end{align}
where \eqref{equ_hessian_lower_loose_symm} follows from the fact that $\bT_{2}^{\top}\bV \in\mathbb{R}^{r\times r}$ is a symmetric matrix \cite[Theorem 2]{ten1977orthogonal}, \eqref{equ_hessian_lower_loose_invoking_lemma2} arises from the fact $\big\Vert \bT_{2}^{\top}\bV \big\Vert_{\F}^{2} \geq 0$ as well as the assumptions of  Lemma~\ref{lemma_restrict_concen_hessian_neighbor}, and \eqref{equ_hessian_lower_loose_finalbound} is based on the fact $\left\Vert \bX^{\natural} \right\Vert \ge \sigma_{r}(\bX^{\natural} ) $. Combining \eqref{equ_hessian_lower_loose_finalbound} with Lemma \ref{lemma_restrict_concen_hessian_lower_looseconst}, we establish the lower bound \eqref{eq:lower_bound}.
 

To prove the upper bound \eqref{eq:upper_bound} asserted in the lemma, we make the observation that the Hessian in \eqref{eq:hessian} satisfies
\begin{align}
&\left\Vert  \nabla^{2}f(\bX)\right\Vert \nonumber\\
& = \left\Vert   \frac{1}{m}\sum_{i=1}^{m} \left[ \left( \Vert \ba_{i}^{\top}\bX \Vert_{2}^{2} -  \Vert \ba_{i}^{\top}\bX^{\natural} \Vert_{2}^{2} \right)\bI_{r} + 2\bX^{\top}\ba_{i}\ba_{i}^{\top}\bX \right] \otimes \left( \ba_{i}\ba_{i}^{\top} \right) \right\Vert \nonumber\\
& \le \left\Vert   \frac{1}{m}\sum_{i=1}^{m} \left[  \left|  \ba_{i}^{\top} \left(\bX + \bX^{\natural}\right) \left(\bX - \bX^{\natural}\right)^{\top} \ba_{i} \right|  \bI_{r} + 2 \big\Vert\ba_{i}^{\top}\bX \big\Vert_{2}^{2}\bI_{r} \right] \otimes \left( \ba_{i}\ba_{i}^{\top} \right) \right\Vert \nonumber\\
&\le \left\Vert   \frac{1}{m}\sum_{i=1}^{m} \left[ \left(  \Vert \ba_{i}^{\top}\bX \Vert_{2} +  \Vert \ba_{i}^{\top}\bX^{\natural} \Vert_{2} \right) \cdot \left\Vert \ba_{i}^{\top}\big(\bX-\bX^{\natural}\big)\right\Vert_{2}    + 2\big\Vert\ba_{i}^{\top}\bX\big\Vert_{2}^{2}  \right]     \ba_{i}\ba_{i}^{\top}   \right\Vert \label{equ_hessian_upper_kron_spectral}\\
& = \Bigg\Vert   \frac{1}{m}\sum_{i=1}^{m} \left[ \left( \Vert \ba_{i}^{\top}\bX\Vert_{2} + \Vert \ba_{i}^{\top}\bX^{\natural}\Vert_{2} \right) \cdot \left\Vert \ba_{i}^{\top}\big(\bX-\bX^{\natural}\big)\right\Vert_{2}   + 2 \left(\big\Vert\ba_{i}^{\top}\bX\big\Vert_{2}^{2} - \big\Vert\ba_{i}^{\top}\bX^{\natural}\big\Vert_{2}^{2} \right) \right]  \left( \ba_{i}\ba_{i}^{\top} \right)  \nonumber\\
& \quad\quad + \frac{1}{m}\sum_{i=1}^{m} 2\big\Vert\ba_{i}^{\top}\bX^{\natural}\big\Vert_{2}^{2}  \left( \ba_{i}\ba_{i}^{\top} \right) -  2 \left(\big\Vert\bX^{\natural}\big\Vert_{\F}^{2}\bI_{n} + 2\bX^{\natural}\bX^{\natural \top}\right) + 2 \left(\big\Vert\bX^{\natural}\big\Vert_{\F}^{2}\bI_{n} + 2\bX^{\natural}\bX^{\natural \top}\right) \Bigg\Vert \nonumber\\
& \le \underbrace{\left\Vert  \frac{3}{m}\sum_{i=1}^{m} \left(  \Vert \ba_{i}^{\top}\bX \Vert_{2} +  \Vert \ba_{i}^{\top}\bX^{\natural} \Vert_{2} \right) \cdot \left\Vert \ba_{i}^{\top}\big(\bX-\bX^{\natural}\big)\right\Vert_{2}   \left( \ba_{i}\ba_{i}^{\top} \right) \right\Vert}_{:=B_1} \nonumber\\
& \ \  + \underbrace{2\left\Vert \frac{1}{m}\sum_{i=1}^{m}\left\Vert\ba_{i}^{\top}\bX^{\natural}\right\Vert_{2}^{2}  \left( \ba_{i}\ba_{i}^{\top} \right) -  \big\Vert\bX^{\natural}\big\Vert_{\F}^{2}\bI_{n} - 2\bX^{\natural}\bX^{\natural \top}   \right\Vert }_{:=B_2}+ \underbrace{2 \left\Vert  \big\Vert\bX^{\natural}\big\Vert_{\F}^{2}\bI_{n} + 2\bX^{\natural}\bX^{\natural \top}   \right\Vert}_{:=B_3} ,  \label{eq:B1-B2-B3}
\end{align}
where \eqref{equ_hessian_upper_kron_spectral} follows from the fact $\left\Vert \bI \otimes \bA\right\Vert =  \left\Vert \bA\right\Vert$. It is seen from Lemma~\ref{lemma_intial_weighted_mat_concen} that 
\begin{equation*}
B_2 \leq \delta \big\Vert\bX^{\natural}\big\Vert_{\F}^{2} \le 0.02 \sigma_{r}^{2}\big(\bX^{\natural}\big),
\end{equation*}
when setting $\delta \le 0.02 \frac{\sigma_{r}^{2}\big(\bX^{\natural}\big)}{\big\Vert\bX^{\natural}\big\Vert_{\F}^{2}}$. Moreover, it is straightforward to check that
\begin{equation*}
B_3\leq 6\big\Vert\bX^{\natural}\big\Vert_{\F}^{2}.
\end{equation*} 
With regards to the first term $B_1$, note that by Lemma~\ref{lemma_a_log_tight_concen} and \eqref{eq:def_incoherence}, we can bound
\begin{equation*}
 \left\Vert \ba_{i}^{\top}\bX \right\Vert_{2}  \leq  \left\Vert \ba_{i}^{\top}\bX^{\natural} \right\Vert_{2} + \left \Vert \ba_{i}^{\top}(\bX-\bX^{\natural}) \right\Vert_{2} \leq 5.86 \sqrt{\log{n}} \big\Vert \bX^{\natural} \big\Vert_{\F} +  \frac{1}{24} \sqrt{\log{n}}  \cdot \frac{\sigma_{r}^{2}\left(\bX^{\natural}\right)}{\left\Vert \bX^{\natural} \right\Vert_{\F}} 
 \end{equation*}
 for $1\leq i\leq m$, and therefore,
\begin{align}
B_1 & \le 1.471 \sigma_{r}^{2}\big(\bX^{\natural}\big)  \log{n} \left\Vert  \frac{1}{m}\sum_{i=1}^{m}  \ba_{i}\ba_{i}^{\top} \right\Vert   \le 1.48 \sigma_{r}^{2}\big(\bX^{\natural}\big) \log{n}    ,\label{equ_hessian_upper_aaT_con}
\end{align}
where the last line follows from Lemma~\ref{lemma_aaT_operator_concentration}. The proof is then finished by combining \eqref{eq:B1-B2-B3} with the preceding bounds on $B_1$, $B_2$ and $B_3$.

\section{Proof of Lemma~\ref{lemma_restrict_concen_hessian_lower_looseconst}}\label{proof_lemma_restrict_concen_hessian_lower_looseconst}

Without loss of generality, we assume $\left\Vert \bV \right\Vert_{\F} = 1$. Write
\begin{align}
&\mathrm{vec}\left(\bV\right)^{\top} \nabla^{2}f(\bX) \mathrm{vec}\left(\bV\right) \nonumber \\
& = \frac{1}{m}\sum_{i=1}^{m} \mathrm{vec}\left(\bV\right)^{\top} \left[ \left[ \left(\left\Vert \ba_{i}^{\top}\bX\right\Vert_{2}^{2} - y_{i}\right)\bI_{r} + 2\bX^{\top}\ba_{i}\ba_{i}^{\top}\bX \right] \otimes \left( \ba_{i}\ba_{i}^{\top} \right) \right] \mathrm{vec}\left(\bV\right) \nonumber\\
& = \frac{1}{m}\sum_{i=1}^{m} \left(\left\Vert \ba_{i}^{\top}\bX\right\Vert_{2}^{2} - y_{i}\right) \mathrm{vec}\left(\bV\right)^{\top}\mathrm{vec}\left(\ba_{i}\ba_{i}^{\top}\bV\right)+ \frac{1}{m}\sum_{i=1}^{m} \mathrm{vec}\left(\bV\right)^{\top} \mathrm{vec}\left(2 \ba_{i}\ba_{i}^{\top}\bV\bX^{\top}\ba_{i}\ba_{i}^{\top}\bX\right)  \nonumber\\
	& = \frac{1}{m}\sum_{i=1}^{m} \left[ \left(\left\Vert \ba_{i}^{\top}\bX\right\Vert_{2}^{2} - \left\Vert \ba_{i}^{\top}\bX^{\natural}\right\Vert_{2}^{2}\right) \left\Vert \ba_{i}^{\top}\bV\right\Vert_{2}^{2} + 2 \left(\ba_{i}^{\top}\bX\bV^{\top}\ba_{i}\right)^{2} \right].  \label{eq:p-original}
\end{align}
In what follows, we let $\bX=\bX^{\natural} + t\frac{\sigma_{r}^{2}\left(\bX^{\natural}\right)}{\left\Vert \bX^{\natural} \right\Vert_{\F}}\bH$ with $t \le 1/24$ and $\left\Vert \bH \right\Vert_{\F} = 1$ which immediately obeys $\left\Vert \bX - \bX^{\natural} \right\Vert_{\F} \le \frac{1}{24} \frac{\sigma_{r}^{2}\left(\bX^{\natural}\right)}{\left\Vert \bX^{\natural} \right\Vert_{\F}}$, and express the right-hand side of \eqref{eq:p-original} as
\begin{align}
&p\left(\bV, \bH, t\right) \nonumber \\ 
& := \underbrace{\frac{1}{m}\sum_{i=1}^{m}  \left[ \left\Vert \ba_{i}^{\top}\bX\right\Vert_{2}^{2} \left\Vert \ba_{i}^{\top}\bV\right\Vert_{2}^{2} + 2 \left(\ba_{i}^{\top}\bX\bV^{\top}\ba_{i}\right)^{2} \right] }_{:=q\left(\bV, \bH, t\right)} - \frac{1}{m}\sum_{i=1}^{m} \left\Vert \ba_{i}^{\top}\bX^{\natural}\right\Vert_{2}^{2} \left\Vert \ba_{i}^{\top}\bV\right\Vert_{2}^{2}.  \label{eq:def_pVH}
\end{align}
The aim is thus to control $p\left(\bV, \bH, t\right)$ for all matrices satisfying $\left\Vert \bH\right\Vert_{\F} = 1$ and $\left\Vert \bV\right\Vert_{\F} = 1$, and for all $t$ obeying $t\le {1}/{24}$.
 
We first bound the second term in \eqref{eq:def_pVH}. Let $\bV = [\bv_{1}, \bv_{2}, \cdots, \bv_{r}]$, then by Lemma~\ref{lemma_intial_weighted_mat_concen}, 
\begin{align*}
&\left\vert \frac{1}{m} \sum_{i=1}^{m} \big\Vert \ba_{i}^{\top}\bX^{\natural}\big\Vert_{2}^{2} \big\Vert \ba_{i}^{\top}\bV\big\Vert_{2}^{2} - \big\Vert \bX^{\natural}\big\Vert_{\F}^{2}\big\Vert \bV\big\Vert_{\F}^{2} - 2 \big\Vert \bX^{\natural \top}\bV \big\Vert_{\F}^{2} \right\vert \\
& = \left\vert \frac{1}{m} \sum_{i=1}^{m} \big\Vert \ba_{i}^{\top}\bX^{\natural}\big\Vert_{2}^{2}  \sum_{k=1}^{r} \left( \ba_{i}^{\top}\bv_{k}\right)^{2} - \big\Vert \bX^{\natural}\big\Vert_{\F}^{2} \sum_{k=1}^{r} \left\Vert \bv_{k}\right\Vert_{2}^{2} - 2 \sum_{k=1}^{r} \big\Vert \bX^{\natural \top}\bv_{k}\big\Vert_{2}^{2} \right\vert\\
& \le \sum_{k=1}^{r}   \left\vert \frac{1}{m} \sum_{i=1}^{m} \big\Vert \ba_{i}^{\top}\bX^{\natural}\big\Vert_{2}^{2}  \left( \ba_{i}^{\top}\bv_{k}\right)^{2} - \big\Vert \bX^{\natural}\big\Vert_{\F}^{2} \left\Vert \bv_{k}\right\Vert_{2}^{2} - 2 \big\Vert \bX^{\natural \top}\bv_{k}\big\Vert_{2}^{2} \right\vert\\
& = \sum_{k=1}^{r}   \left\vert \bv_{k}^{\top} \left( \frac{1}{m} \sum_{i=1}^{m} \big\Vert \ba_{i}^{\top}\bX^{\natural}\big\Vert_{2}^{2} \ba_{i}\ba_{i}^{\top} - \big\Vert \bX^{\natural}\big\Vert_{\F}^{2} - 2 \bX^{\natural}\bX^{\natural \top}  \right) \bv_{k} \right\vert\\
& \le \sum_{k=1}^{r} \left\Vert \bv_{k} \right\Vert_{2}^{2} \left\Vert  \frac{1}{m} \sum_{i=1}^{m} \big\Vert \ba_{i}^{\top}\bX^{\natural}\big\Vert_{2}^{2} \ba_{i}\ba_{i}^{\top} - \big\Vert \bX^{\natural}\big\Vert_{\F}^{2} - 2 \bX^{\natural}\bX^{\natural \top}  \right\Vert  \\
& \le \delta \big\Vert \bX^{\natural} \big\Vert_{\F}^{2}  \sum_{k=1}^{r} \left\Vert \bv_{k} \right\Vert_{2}^{2}   
 =  \delta \big\Vert \bX^{\natural} \big\Vert_{\F}^{2} \left\Vert \bV \right\Vert_{\F}^{2}.
\end{align*}
By setting $\delta \le \frac{1}{24}\frac{ \sigma_{r}^{2}\left(\bX^{\natural}\right)}{ \left\Vert \bX^{\natural}\right\Vert_{\F}^{2}}$, we see that with probability at least  $1 - c_{1} r n^{-13}$,
\begin{equation}\label{equ_hessian_loose_secondterm}
 \frac{1}{m} \sum_{i=1}^{m} \big\Vert \ba_{i}^{\top}\bX^{\natural}\big\Vert_{2}^{2} \big\Vert \ba_{i}^{\top}\bV \big\Vert_{2}^{2} \le  \big\Vert \bX^{\natural} \big\Vert_{\F}^{2} \big\Vert \bV \big\Vert_{\F}^{2} + 2 \big\Vert \bX^{\natural \top}\bV\big\Vert_{\F}^{2} + \frac{1}{24} \sigma_{r}^{2}\big(\bX^{\natural}\big)  \left\Vert \bV \right\Vert_{\F}^{2},
\end{equation}
holds simultaneously for all matrices $\bV$, as long as $m\gtrsim  \frac{\left\Vert\bX^{\natural}\right\Vert_{\F}^{4}}{\sigma_{r}^{4}\left(\bX^{\natural}\right)} n \log{n}$.

Next, we turn to  the first term $q\left(\bV, \bH, t\right)$ in \eqref{eq:def_pVH}, and we need to accommodate all matrices satisfying $\left\Vert \bH\right\Vert_{\F} = 1$ and $\left\Vert \bV\right\Vert_{\F} = 1$, and all scalars obeying $t\le {1}/{24}$. The strategy is that we first establish the bound of $q\left(\bV, \bH, t\right)$ for any fixed $\bH$, $\bV$ and $t$, and then extend the result to a uniform bound for all $\bH$, $\bV$ and $t$ by covering arguments.

\subsection{Bound with Fixed Matrices and Scalar}

Recall that
\begin{align*}
q\left(\bV, \bH, t\right)& =  \frac{1}{m}\sum_{i=1}^{m} \underbrace{\left[ \big\Vert \ba_{i}^{\top}\bX \big\Vert_{2}^{2} \big\Vert \ba_{i}^{\top}\bV \big\Vert_{2}^{2} + 2 \big(\ba_{i}^{\top}\bX\bV^{\top}\ba_{i}\big)^{2} \right]}_{:=G_{i}}.
\end{align*}
%
We will start by assuming that $\bX$ and $\bV$ are both fixed and statistically independent of $\{\ba_i\}_{i=1}^{m}$. 
In view of Lemma~\ref{lemma_expectation_gaussian}, 
\begin{align}
\mathbb{E}\left[ G_{i} \right]
& = \mathbb{E}\left[ \big\Vert \ba_{i}^{\top}\bX\big\Vert_{2}^{2} \big\Vert \ba_{i}^{\top}\bV\big\Vert_{2}^{2} \right] + 2 \mathbb{E}\left[ \big(\ba_{i}^{\top}\bX\bV^{\top}\ba_{i}\big)^{2}\right] \nonumber\\
& = \left\Vert\bX\right\Vert_{\F}^{2} \left\Vert\bV\right\Vert_{\F}^{2} + 2 \big\Vert\bX^{\top}\bV\big\Vert_{\F}^{2} + 2 \left( \mathrm{Tr}\big(\bX^{\top}\bV\big)\right)^{2} + 2 \big\Vert \bX\bV^{\top}\big\Vert_{\F}^{2} + 2 \mathrm{Tr}\big(\bX^{\top}\bV\bX^{\top}\bV\big) \nonumber\\
& \le \left\Vert\bX\right\Vert_{\F}^{2} \left\Vert\bV\right\Vert_{\F}^{2} + 2\left\Vert\bX\right\Vert^{2} \left\Vert\bV\right\Vert_{\F}^{2} + 2\left\Vert\bX\right\Vert_{\F}^{2} \left\Vert\bV\right\Vert_{\F}^{2} + 2 \left\Vert\bX\right\Vert^{2} \left\Vert\bV\right\Vert_{\F}^{2} + 2 \left\Vert\bX\right\Vert^{2} \left\Vert\bV\right\Vert_{\F}^{2} \nonumber\\
& \le 9 \left\Vert \bX\right\Vert_{\F}^{2} \left\Vert\bV\right\Vert_{\F}^{2}
= 9 \left\Vert  \bX^{\natural} + t\frac{\sigma_{r}^{2}\left(\bX^{\natural}\right)}{\left\Vert \bX^{\natural} \right\Vert_{\F}}\bH \right\Vert_{\F}^{2} \label{equ_expectation_G_invokex} \\
&\le 18  \left( \big\Vert\bX^{\natural}\big\Vert_{\F}^{2} + t^{2}\frac{\sigma_{r}^{4}\left(\bX^{\natural}\right)}{\left\Vert \bX^{\natural} \right\Vert_{\F}^{2}} \left\Vert\bH\right\Vert_{\F}^{2} \right)
 \le 18.002  \big\Vert\bX^{\natural}\big\Vert_{\F}^{2} , \label{equ_expectation_G_invokeht}
\end{align}
where \eqref{equ_expectation_G_invokex} follows $\|\bV\|_{\F}=1$ and $\bX=\bX^{\natural} + t\frac{\sigma_{r}^{2}\left(\bX^{\natural}\right)}{\left\Vert \bX^{\natural} \right\Vert_{\F}}\bH$, and \eqref{equ_expectation_G_invokeht} arises from the calculations with  $\|\bH\|_{\F}=1$ and $t \le 1/24$. Therefore, if we define $T_{i} = \mathbb{E}\left[ G_{i} \right] - G_{i}$, we have $\mathbb{E}\left[T_{i}\right] = 0$ and 
\begin{equation*}
T_{i} \le \mathbb{E}\left[G_{i}\right] \le 18.002  \big \Vert\bX^{\natural}\big\Vert_{\F}^{2},
\end{equation*} 
due to $G_{i} \ge 0$.  In addition,  
\begin{align}
	& \mathbb{E}\left[ T_{i}^{2} \right] = \mathbb{E}\left[ G_{i}^{2} \right] - \left( \mathbb{E}\left[G_{i}\right] \right)^{2} 
 \le \mathbb{E}\left[ G_{i}^{2} \right] \nonumber\\
& = \mathbb{E}\left[ \left( \big\Vert \ba_{i}^{\top}\bX\big\Vert_{2}^{2} \big\Vert \ba_{i}^{\top}\bV\big\Vert_{2}^{2} + 2 \big(\ba_{i}^{\top}\bX\bV^{\top}\ba_{i}\big)^{2} \right)^{2} \right] \nonumber\\
	& = \mathbb{E}\left[  \big\Vert \ba_{i}^{\top}\bX\big\Vert_{2}^{4} \big\Vert \ba_{i}^{\top}\bV\big\Vert_{2}^{4} \right] + 4 \mathbb{E}\left[ \big(\ba_{i}^{\top}\bX\bV^{\top}\ba_{i}\big)^{4} \right] + 4 \mathbb{E}\left[ \big(\ba_{i}^{\top}\bX\bV^{\top}\ba_{i}\big)^{2}  \big\Vert \ba_{i}^{\top}\bX\big\Vert_{2}^{2} \big\Vert \ba_{i}^{\top}\bV\big\Vert_{2}^{2}\right] \nonumber\\
& \le 9  \mathbb{E}\left[  \big\Vert \ba_{i}^{\top}\bX\big\Vert_{2}^{4} \big\Vert \ba_{i}^{\top}\bV\big\Vert_{2}^{4} \right]  \label{equ_tsqure_upperbound_cauchy}\\
& \le  9  \sqrt{ \mathbb{E}\left[  \left\Vert \ba_{i}^{\top}\bX\right\Vert_{2}^{8}\right]  \mathbb{E}\left[   \left\Vert \ba_{i}^{\top}\bV\right\Vert_{2}^{8} \right]}  \label{equ_tsqure_upperbound_holder}\\
& \le 9c_{4}  \left\Vert\bX \right\Vert_{\F}^{4} \left\Vert\bV \right\Vert_{\F}^{4} = 9c_{4} \big\Vert\bX \big\Vert_{\F}^{4} \label{equ_tsqure_upperbound_gaussian}\\
&  = 9c_{4} \left\Vert \bX^{\natural} + t\frac{\sigma_{r}^{2}\left(\bX^{\natural}\right)}{\left\Vert \bX^{\natural} \right\Vert_{\F}}\bH \right\Vert_{\F}^{4} \lesssim \big\Vert \bX^{\natural} \big\Vert_{\F}^{4}, \nonumber
\end{align}
where \eqref{equ_tsqure_upperbound_cauchy} follows from the Cauchy-Schwarz inequality, \eqref{equ_tsqure_upperbound_holder} comes from the H\"{o}lder's inequality, and \eqref{equ_tsqure_upperbound_gaussian} is a consequence of Lemma~\ref{lemma_expectation_gaussian}. Apply Lemma~\ref{lemma_oneside_tail_bound} to arrive at 
\begin{equation}
\mathbb{P}\left( \frac{1}{m} \sum_{i=1}^{m} T_{i} \ge \frac{1}{24}  \sigma_{r}^{2}\left(\bX^{\natural}\right)  \right) \le \exp\left(-c\frac{m\sigma_{r}^{4}\left(\bX^{\natural}\right)}{\left\Vert\bX^{\natural}\right\Vert_{\F}^{4}}\right),
\end{equation}
which further leads to
\begin{align}
&q\left(\bV, \bH, t\right)  =  \frac{1}{m}\sum_{i=1}^{m} G_{i}= \mathbb{E}\left[G_{i}\right]  -  \frac{1}{m}\sum_{i=1}^{m} T_{i} \nonumber\\
& \ge \mathbb{E}\left[G_{i}\right] -  \frac{1}{24} \sigma_{r}^{2}\big(\bX^{\natural}\big) \nonumber\\
& =  \left\Vert\bX\right\Vert_{\F}^{2} \left\Vert\bV\right\Vert_{\F}^{2} + 2 \big\Vert\bX^{\top}\bV\big\Vert_{\F}^{2} + 2 \left( \mathrm{Tr}\big(\bX^{\top}\bV\big)\right)^{2} + 2 \big\Vert \bX\bV^{\top}\big\Vert_{\F}^{2} + 2 \mathrm{Tr}\left(\bX^{\top}\bV\bX^{\top}\bV\right) -  \frac{1}{24} \sigma_{r}^{2}\left(\bX^{\natural}\right)  \nonumber\\
& \ge  \left\Vert\bX\right\Vert_{\F}^{2} \left\Vert\bV\right\Vert_{\F}^{2} + 2 \big\Vert\bX^{\top}\bV\big\Vert_{\F}^{2} + 2 \big\Vert \bX\bV^{\top}\big\Vert_{\F}^{2} + 2 \mathrm{Tr}\left(\bX^{\top}\bV\bX^{\top}\bV\right) -  \frac{1}{24} \sigma_{r}^{2}\left(\bX^{\natural}\right) . \label{equ_q_interim_lowerbound}
\end{align}

Substituting $\bX=\bX^{\natural} + t\frac{\sigma_{r}^{2}\left(\bX^{\natural}\right)}{\left\Vert \bX^{\natural} \right\Vert_{\F}}\bH$ for $\bX$, and using the facts $\left\Vert \bH\right\Vert_{\F} = 1$, $\left\Vert \bV\right\Vert_{\F} = 1$ and $t \le 1/24$, we can calculate the following bounds:
\begin{align*}
\left\Vert\bX\right\Vert_{\F}^{2} 
& =  \big\Vert \bX^{\natural} \big\Vert_{\F}^{2} +  t^2 \frac{\sigma_{r}^{4}\left(\bX^{\natural}\right)}{\left\Vert \bX^{\natural} \right\Vert_{\F}^2 } \left\Vert\bH\right\Vert_{\F}^{2}  + 2t\frac{\sigma_{r}^{2}\left(\bX^{\natural}\right)}{\left\Vert \bX^{\natural} \right\Vert_{\F}} \mathrm{Tr}\left(   \bX^{\natural \top}\bH  \right)\\
& \ge \big\Vert \bX^{\natural} \big\Vert_{\F}^{2} - 2  t\frac{\sigma_{r}^{2}\left(\bX^{\natural}\right)}{\left\Vert \bX^{\natural} \right\Vert_{\F}}  \big\Vert\bX^{\natural}\big\Vert_{\F}\left\Vert\bH\right\Vert_{\F} \ge \big\Vert \bX^{\natural} \big\Vert_{\F}^{2} - \frac{1}{12} \sigma_{r}^{2}\left(\bX^{\natural}\right);\\
 \big\Vert\bX^{\top}\bV\big\Vert_{\F}^{2}
 & =  \left\Vert \bX^{\natural \top}\bV\right\Vert_{\F}^{2} + t^2 \frac{\sigma_{r}^{4}\left(\bX^{\natural}\right)}{\left\Vert \bX^{\natural} \right\Vert_{\F}^2} \left\Vert \bH^{\top}\bV\right\Vert_{\F}^{2} + 2t\frac{\sigma_{r}^{2}\left(\bX^{\natural}\right)}{\left\Vert \bX^{\natural} \right\Vert_{\F}} \mathrm{Tr}\left(    \bV^{\top}\bH\bX^{\natural \top}\bV\right) \\
 & \ge  \left\Vert \bX^{\natural \top}\bV\right\Vert_{\F}^{2} -  2 t\frac{\sigma_{r}^{2}\left(\bX^{\natural}\right)}{\left\Vert \bX^{\natural} \right\Vert_{\F}} \big\Vert\bX^{\natural}\big\Vert \left\Vert\bH\right\Vert \left\Vert\bV\right\Vert_{\F}^{2} \ge \left\Vert \bX^{\natural \top}\bV\right\Vert_{\F}^{2} -  \frac{1}{12} \sigma_{r}^{2}\left(\bX^{\natural}\right);\\
 \big\Vert \bX\bV^{\top}\big\Vert_{\F}^{2}
 & =  \left\Vert \bX^{\natural}\bV^{\top}\right\Vert_{\F}^{2} +  t^{2}\frac{\sigma_{r}^{4}\left(\bX^{\natural}\right)}{\left\Vert \bX^{\natural} \right\Vert_{\F}^{2}} \left\Vert \bH\bV^{\top}\right\Vert_{\F}^{2} + 2 t\frac{\sigma_{r}^{2}\left(\bX^{\natural}\right)}{\left\Vert \bX^{\natural} \right\Vert_{\F}} \mathrm{Tr}\left( \bV\bH^{\top}\bX^{\natural}\bV^{\top}\right) \\
 & \ge  \left\Vert \bX^{\natural}\bV^{\top}\right\Vert_{\F}^{2} -  2 t\frac{\sigma_{r}^{2}\left(\bX^{\natural}\right)}{\left\Vert \bX^{\natural} \right\Vert_{\F}} \big\Vert\bX^{\natural}\big\Vert  \left\Vert\bH\right\Vert \left\Vert\bV\right\Vert_{\F}^{2} \ge \left\Vert \bX^{\natural}\bV^{\top}\right\Vert_{\F}^{2} -  \frac{1}{12} \sigma_{r}^{2}\left(\bX^{\natural}\right);\\
 \mathrm{Tr}\left(\bX^{\top}\bV\bX^{\top}\bV\right)
 & =  \mathrm{Tr}\left(\bX^{\natural \top}\bV\bX^{\natural \top}\bV\right) + 2 t\frac{\sigma_{r}^{2}\left(\bX^{\natural}\right)}{\left\Vert \bX^{\natural} \right\Vert_{\F}} \mathrm{Tr}\left( \bH^{\top}\bV\bX^{\natural \top}\bV\right)  + t^{2}\frac{\sigma_{r}^{4}\left(\bX^{\natural}\right)}{\left\Vert \bX^{\natural} \right\Vert_{\F}^{2}} \mathrm{Tr}\left(\bH^{\top}\bV \bH^{\top}\bV\right)\\
 & \ge \mathrm{Tr}\left(\bX^{\natural \top}\bV\bX^{\natural \top}\bV\right)  -  2  t\frac{\sigma_{r}^{2}\left(\bX^{\natural}\right)}{\left\Vert \bX^{\natural} \right\Vert_{\F}} \big\Vert\bX^{\natural}\big\Vert \left\Vert\bH\right\Vert \left\Vert\bV\right\Vert_{\F}^{2} -  t^{2}\frac{\sigma_{r}^{4}\left(\bX^{\natural}\right)}{\left\Vert \bX^{\natural} \right\Vert_{\F}^{2}} \left\Vert\bH\right\Vert^{2}\left\Vert\bV\right\Vert_{\F}^{2}\\
 & \ge \mathrm{Tr}\left(\bX^{\natural \top}\bV\bX^{\natural \top}\bV\right) - \left( \frac{1}{12} + \frac{1}{24^{2}}\right) \sigma_{r}^{2}\left(\bX^{\natural}\right),
\end{align*}
which, combining with \eqref{equ_q_interim_lowerbound}, yields 
\begin{align*}
& q\left(\bV, \bH, t\right) \\
& \ge \big\Vert \bX^{\natural} \big\Vert_{\F}^{2} + 2 \left\Vert \bX^{\natural \top}\bV\right\Vert_{\F}^{2}  + 2 \left\Vert \bX^{\natural}\bV^{\top}\right\Vert_{\F}^{2}  +  2 \mathrm{Tr}\left(\bX^{\natural \top}\bV\bX^{\natural \top}\bV\right)  -  \left(\frac{15}{24} + \frac{1}{12\cdot 24} \right) \sigma_{r}^{2}\left(\bX^{\natural}\right)\\
& \ge  \big\Vert \bX^{\natural} \big\Vert_{\F}^{2} + 2 \left\Vert \bX^{\natural \top}\bV\right\Vert_{\F}^{2}   +  2 \mathrm{Tr}\left(\bX^{\natural \top}\bV\bX^{\natural \top}\bV\right)  + 2  \sigma_{r}^{2}\left(\bX^{\natural}\right) -  \left(\frac{15}{24} + \frac{1}{12\cdot 24} \right) \sigma_{r}^{2}\left(\bX^{\natural}\right)\\
& \ge  \big\Vert \bX^{\natural} \big\Vert_{\F}^{2} + 2 \left\Vert \bX^{\natural \top}\bV\right\Vert_{\F}^{2}   +  2 \mathrm{Tr}\left(\bX^{\natural \top}\bV\bX^{\natural \top}\bV\right)  + 1.371  \sigma_{r}^{2}\left(\bX^{\natural}\right).
\end{align*}

\subsection{Covering Arguments}

Since we have obtained a lower bound on $q\left(\bV, \bH, t\right) $ for  fixed $\bV$, $\bH$ and $t$, we now move on to extending it to a uniform bound that covers all $\bV$, $\bH$ and $t$ simultaneously. 
Towards this, we will 
invoke the $\epsilon$-net covering arguments for all $\bV$, $\bH$ and $t$, respectively, and will rely on the fact $\max_{1\le i\le m}\left\Vert\ba_{i}\right\Vert_{2} \le \sqrt{6 n}$ asserted in Lemma~\ref{lemma_a_sqrtroot_concen}. For notational convenience, we define 
\begin{align*}
g\left(\bV, \bH, t\right) 
& = q\left(\bV, \bH, t\right)\\
& \quad -   \big\Vert \bX^{\natural} \big\Vert_{\F}^{2} - 2 \big\Vert \bX^{\natural \top}\bV\big\Vert_{\F}^{2}   -  2 \mathrm{Tr}\left(\bX^{\natural \top}\bV\bX^{\natural \top}\bV\right)  - 1.371  \sigma_{r}^{2}\left(\bX^{\natural}\right).
\end{align*}

First, consider the $\epsilon$-net covering argument for $\bV$. Suppose $\bV_{1}$ and $\bV_{2}$ are such that $\left\Vert \bV_{1} \right\Vert_{\F} = 1$, $\left\Vert \bV_{2} \right\Vert_{\F} = 1$, and $\left\Vert \bV_{1} - \bV_{2} \right\Vert_{\F} \le \epsilon$. Then, since
\begin{align*}
\left\vert  \left\Vert\bX^{\natural \top}\bV_{1}\right\Vert_{\F}^{2} -  \left\Vert\bX^{\natural \top}\bV_{2}\right\Vert_{\F}^{2} \right\vert \le \left( \left\Vert\bX^{\natural \top}\bV_{1}\right\Vert_{\F} +  \left\Vert\bX^{\natural \top}\bV_{2}\right\Vert_{\F} \right)  \left\Vert\bX^{\natural \top}\left(\bV_{1}-\bV_{2}\right)\right\Vert_{\F} \le 2\big\Vert \bX^{\natural}\big\Vert^{2} \epsilon,
\end{align*}
and
\begin{align*}
&\left\vert \mathrm{Tr}\left(\bX^{\natural \top}\bV_{1}\bX^{\natural \top}\bV_{1}\right) - \mathrm{Tr}\left(\bX^{\natural \top}\bV_{2}\bX^{\natural \top}\bV_{2}\right) \right\vert\\
&\le \left\vert \mathrm{Tr}\left(\bX^{\natural \top}\bV_{1}\bX^{\natural \top}\bV_{1}\right) - \mathrm{Tr}\left(\bX^{\natural \top}\bV_{1}\bX^{\natural \top}\bV_{2}\right) \right\vert + \left\vert \mathrm{Tr}\left(\bX^{\natural \top}\bV_{1}\bX^{\natural \top}\bV_{2}\right) - \mathrm{Tr}\left(\bX^{\natural \top}\bV_{2}\bX^{\natural \top}\bV_{2}\right) \right\vert \\
& \le \big\Vert \bX^{\natural}\big\Vert^{2} \left\Vert \bV_{1} \right\Vert_{\F} \left\Vert \bV_{1} -\bV_{2} \right\Vert_{\F} +  \big\Vert \bX^{\natural}\big\Vert^{2} \left\Vert \bV_{2} \right\Vert_{\F} \left\Vert \bV_{1} -\bV_{2} \right\Vert_{\F} \le 2\big\Vert \bX^{\natural}\big\Vert^{2}\epsilon,
\end{align*}
we have
\begin{align*}
&\left\vert g\left(\bV_{1}, \bH, t\right) - g\left(\bV_{2}, \bH, t\right) \right\vert\\
& \le \left\vert q\left(\bV_{1}, \bH, t\right) - q\left(\bV_{2}, \bH, t\right) \right\vert + 2\left\vert  \left\Vert\bX^{\natural \top}\bV_{1}\right\Vert_{\F}^{2} -  \left\Vert\bX^{\natural \top}\bV_{2}\right\Vert_{\F}^{2} \right\vert  \\
&\quad + 2 \left\vert \mathrm{Tr}\left(\bX^{\natural \top}\bV_{1}\bX^{\natural \top}\bV_{1}\right) - \mathrm{Tr}\left(\bX^{\natural \top}\bV_{2}\bX^{\natural \top}\bV_{2}\right) \right\vert\\
& \le \left\vert \frac{1}{m}\sum_{i=1}^{m} \left[ \left\Vert \ba_{i}^{\top}\bX\right\Vert_{2}^{2} \left\Vert \ba_{i}^{\top}\bV_{1}\right\Vert_{2}^{2} + 2 \left(\ba_{i}^{\top}\bX\bV_{1}^{\top}\ba_{i}\right)^{2} \right] - \frac{1}{m}\sum_{i=1}^{m} \left[ \left\Vert \ba_{i}^{\top}\bX\right\Vert_{2}^{2} \left\Vert \ba_{i}^{\top}\bV_{2}\right\Vert_{2}^{2} + 2 \left(\ba_{i}^{\top}\bX\bV_{2}^{\top}\ba_{i}\right)^{2} \right]\right\vert\\
& \quad + 8 \big\Vert \bX^{\natural}\big\Vert^{2}\epsilon  \\
& \le  \frac{1}{m}\sum_{i=1}^{m}  \left\vert \left\Vert \ba_{i}^{\top}\bX\right\Vert_{2}^{2} \left\Vert \ba_{i}^{\top}\bV_{1}\right\Vert_{2}^{2} - \left\Vert \ba_{i}^{\top}\bX\right\Vert_{2}^{2} \left\Vert \ba_{i}^{\top}\bV_{2}\right\Vert_{2}^{2} \right\vert  +    \frac{2}{m}\sum_{i=1}^{m} \left\vert \left(\ba_{i}^{\top}\bX\bV_{1}^{\top}\ba_{i}\right)^{2} - \left(\ba_{i}^{\top}\bX\bV_{2}^{\top}\ba_{i}\right)^{2} \right\vert  + 8 \big\Vert \bX^{\natural}\big\Vert^{2}\epsilon \\
& \le \frac{1}{m}\sum_{i=1}^{m}  \left\Vert \ba_{i}^{\top}\bX\right\Vert_{2}^{2} \cdot  \left( \left\Vert \ba_{i}^{\top}\bV_{1}\right\Vert_{2} + \left\Vert \ba_{i}^{\top}\bV_{2}\right\Vert_{2} \right) \cdot \left\Vert \ba_{i}^{\top}\left(\bV_{1} - \bV_{2}\right) \right\Vert_{2}\\
& \quad + \frac{2}{m}\sum_{i=1}^{m} \left\vert \ba_{i}^{\top}\bX \left(\bV_{1}+\bV_{2}\right)^{\top}\ba_{i} \right\vert \cdot \left\vert \ba_{i}^{\top}\bX\left(\bV_{1}-\bV_{2}\right)^{\top}\ba_{i} \right\vert  + 8 \big\Vert \bX^{\natural}\big\Vert^{2}\epsilon  \\
& \le 6 n \cdot \left\Vert\bX\right\Vert^{2} \cdot 2\sqrt{6 n } \cdot \sqrt{6 n } \cdot \epsilon + 2 \cdot 12 n \cdot \left\Vert\bX\right\Vert \cdot 6 n \cdot \left\Vert \bX \right\Vert \epsilon + 8 \big\Vert \bX^{\natural}\big\Vert^{2}\epsilon \\
& = 216 \epsilon n^{2}  \left\Vert  \bX^{\natural} + t\frac{\sigma_{r}^{2}\left(\bX^{\natural}\right)}{\left\Vert \bX^{\natural} \right\Vert_{\F}}\bH \right\Vert^{2}  + 8 \big\Vert \bX^{\natural}\big\Vert^{2}\epsilon  \\
& \le 432 \epsilon n^{2}  \left(  \big\Vert  \bX^{\natural} \big\Vert^{2} +  t^{2}\frac{\sigma_{r}^{4}\left(\bX^{\natural}\right)}{\left\Vert \bX^{\natural} \right\Vert_{\F}^{2}} \left\Vert  \bH \right\Vert^{2}\right)  + 8 \big\Vert \bX^{\natural}\big\Vert^{2}\epsilon \\
& \le \left(432.75 n^{2}  + 8\right)  \epsilon \big\Vert  \bX^{\natural} \big\Vert^{2}
 \le \frac{1}{24} \sigma_{r}^{2}\left(\bX^{\natural}\right),
\end{align*}
as long as $\epsilon = \frac{\sigma_{r}^{2}\left(\bX^{\natural}\right)}{10584 n^{2} \left\Vert  \bX^{\natural} \right\Vert^{2}}$. Based on Lemma~\ref{lemma_covering_net}, the cardinality of this $\epsilon$-net will be 
$$\left(\frac{9}{\epsilon}\right)^{(n+r+1)r} = \left(\frac{9 \cdot 10584 n^{2} \left\Vert  \bX^{\natural} \right\Vert^{2}}{\sigma_{r}^{2}\left(\bX^{\natural}\right)}\right)^{(n+r+1)r} \le \exp{\left(cnr \log{\left(n \kappa\right)}\right)}.$$

Secondly, consider the $\epsilon$-net covering argument for $\bH$. Suppose $\bH_{1}$ and $\bH_{2}$ obey $\left\Vert \bH_{1} \right\Vert_{\F} = 1$, $\left\Vert \bH_{2} \right\Vert_{\F} = 1$, and $\left\Vert \bH_{1} - \bH_{2} \right\Vert_{\F} \le \epsilon$. Then one has
\begin{align*}
&\left\vert g\left(\bV, \bH_{1}, t\right) - g\left(\bV, \bH_{2}, t\right) \right\vert\\
& = \left\vert q\left(\bV, \bH_{1}, t\right) - q\left(\bV, \bH_{2}, t\right) \right\vert \\
& = \Bigg\vert \frac{1}{m}\sum_{i=1}^{m} \left[ \left\Vert \ba_{i}^{\top}\left(\bX^{\natural} + t\frac{\sigma_{r}^{2}\left(\bX^{\natural}\right)}{\left\Vert \bX^{\natural} \right\Vert_{\F}}\bH_{1}\right)\right\Vert_{2}^{2} \left\Vert \ba_{i}^{\top}\bV\right\Vert_{2}^{2} + 2 \left(\ba_{i}^{\top}\left(\bX^{\natural} + t\frac{\sigma_{r}^{2}\left(\bX^{\natural}\right)}{\left\Vert \bX^{\natural} \right\Vert_{\F}}\bH_{1}\right)\bV^{\top}\ba_{i}\right)^{2} \right] \\
& \quad - \frac{1}{m}\sum_{i=1}^{m} \left[ \left\Vert \ba_{i}^{\top}\left(\bX^{\natural} + t\frac{\sigma_{r}^{2}\left(\bX^{\natural}\right)}{\left\Vert \bX^{\natural} \right\Vert_{\F}}\bH_{2}\right)\right\Vert_{2}^{2} \left\Vert \ba_{i}^{\top}\bV\right\Vert_{2}^{2} + 2 \left(\ba_{i}^{\top}\left(\bX^{\natural} + t\frac{\sigma_{r}^{2}\left(\bX^{\natural}\right)}{\left\Vert \bX^{\natural} \right\Vert_{\F}}\bH_{2}\right)\bV^{\top}\ba_{i}\right)^{2} \right] \Bigg\vert\\
& \le  \frac{1}{m}\sum_{i=1}^{m} \left\Vert \ba_{i}^{\top}\bV\right\Vert_{2}^{2} \cdot \left\vert \left\Vert \ba_{i}^{\top}\left(\bX^{\natural} + t\frac{\sigma_{r}^{2}\left(\bX^{\natural}\right)}{\left\Vert \bX^{\natural} \right\Vert_{\F}}\bH_{1}\right)\right\Vert_{2}^{2} - \left\Vert \ba_{i}^{\top}\left(\bX^{\natural} + t\frac{\sigma_{r}^{2}\left(\bX^{\natural}\right)}{\left\Vert \bX^{\natural} \right\Vert_{\F}}\bH_{2}\right)\right\Vert_{2}^{2}  \right\vert\\
& \quad   +    \frac{2}{m}\sum_{i=1}^{m} \left\vert \left(\ba_{i}^{\top}\left(\bX^{\natural} + t\frac{\sigma_{r}^{2}\left(\bX^{\natural}\right)}{\left\Vert \bX^{\natural} \right\Vert_{\F}}\bH_{1}\right)\bV^{\top}\ba_{i}\right)^{2} - \left(\ba_{i}^{\top}\left(\bX^{\natural} + t\frac{\sigma_{r}^{2}\left(\bX^{\natural}\right)}{\left\Vert \bX^{\natural} \right\Vert_{\F}}\bH_{2}\right)\bV^{\top}\ba_{i}\right)^{2} \right\vert   \\
& \le 6 n   \cdot \sqrt{6 n } \cdot t\frac{\sigma_{r}^{2}\left(\bX^{\natural}\right)}{\left\Vert \bX^{\natural} \right\Vert_{\F}}\epsilon \cdot 2\sqrt{6 n } \cdot \frac{25}{24}\big\Vert\bX^{\natural}\big\Vert + 2\cdot 6 n \cdot  t\frac{\sigma_{r}^{2}\left(\bX^{\natural}\right)}{\left\Vert \bX^{\natural} \right\Vert_{\F}}\epsilon \cdot 12 n \cdot \frac{25}{24}\big\Vert\bX^{\natural}\big\Vert\\
& \le \frac{75}{8} \epsilon n^{2} \frac{\sigma_{r}^{2}\left(\bX^{\natural}\right)}{\left\Vert \bX^{\natural} \right\Vert_{\F}}\big\Vert\bX^{\natural}\big\Vert
 \le \frac{1}{24} \sigma_{r}^{2}\left(\bX^{\natural}\right),
\end{align*}
as long as $\epsilon = \frac{1}{225 n^{2}} \cdot \frac{\left\Vert\bX^{\natural}\right\Vert_{\F}}{\left\Vert\bX^{\natural}\right\Vert}$. Based on Lemma~\ref{lemma_covering_net}, the cardinality of this $\epsilon$-net will be 
$$\left(\frac{9}{\epsilon}\right)^{(n+r+1)r} = \left(9 \cdot 225 n^{2} \cdot \frac{\left\Vert\bX^{\natural}\right\Vert}{\left\Vert\bX^{\natural}\right\Vert_{\F}} \right)^{(n+r+1)r} \le \exp{\left(cnr \log{n}\right)}.$$

Finally, consider the $\epsilon$-net covering argument for all $t$, such that $t\le {1}/{24}$. Suppose $t_{1}$ and $t_{2}$ satisfy $t_{1}\le {1}/{24}$, $t_{2}\le {1}/{24}$ and $\left\vert t_{1} - t_{2} \right\vert\le \epsilon$. Then we get 
\begin{align*}
&\left\vert g\left(\bV, \bH, t_{1}\right) - g\left(\bV, \bH, t_{2}\right) \right\vert\\
& = \left\vert q\left(\bV, \bH, t_{1}\right) - q\left(\bV, \bH, t_{2}\right) \right\vert \\
& = \Bigg\vert \frac{1}{m}\sum_{i=1}^{m} \left[ \left\Vert \ba_{i}^{\top}\left(\bX^{\natural} + t_{1}\frac{\sigma_{r}^{2}\left(\bX^{\natural}\right)}{\left\Vert \bX^{\natural} \right\Vert_{\F}}\bH\right)\right\Vert_{2}^{2} \left\Vert \ba_{i}^{\top}\bV\right\Vert_{2}^{2} + 2 \left(\ba_{i}^{\top}\left(\bX^{\natural} + t_{1}\frac{\sigma_{r}^{2}\left(\bX^{\natural}\right)}{\left\Vert \bX^{\natural} \right\Vert_{\F}}\bH\right)\bV^{\top}\ba_{i}\right)^{2} \right] \\
& \quad - \frac{1}{m}\sum_{i=1}^{m} \left[ \left\Vert \ba_{i}^{\top}\left(\bX^{\natural} + t_{2}\frac{\sigma_{r}^{2}\left(\bX^{\natural}\right)}{\left\Vert \bX^{\natural} \right\Vert_{\F}}\bH\right)\right\Vert_{2}^{2} \left\Vert \ba_{i}^{\top}\bV\right\Vert_{2}^{2} + 2 \left(\ba_{i}^{\top}\left(\bX^{\natural} + t_{2}\frac{\sigma_{r}^{2}\left(\bX^{\natural}\right)}{\left\Vert \bX^{\natural} \right\Vert_{\F}}\bH\right)\bV^{\top}\ba_{i}\right)^{2} \right] \Bigg\vert\\
& \le  \frac{1}{m}\sum_{i=1}^{m} \left\Vert \ba_{i}^{\top}\bV\right\Vert_{2}^{2} \cdot \left\vert \left\Vert \ba_{i}^{\top}\left(\bX^{\natural} + t_{1}\frac{\sigma_{r}^{2}\left(\bX^{\natural}\right)}{\left\Vert \bX^{\natural} \right\Vert_{\F}}\bH\right)\right\Vert_{2}^{2} - \left\Vert \ba_{i}^{\top}\left(\bX^{\natural} + t_{2}\frac{\sigma_{r}^{2}\left(\bX^{\natural}\right)}{\left\Vert \bX^{\natural} \right\Vert_{\F}}\bH\right)\right\Vert_{2}^{2}  \right\vert\\
& \quad   +    \frac{2}{m}\sum_{i=1}^{m} \left\vert \left(\ba_{i}^{\top}\left(\bX^{\natural} + t_{1}\frac{\sigma_{r}^{2}\left(\bX^{\natural}\right)}{\left\Vert \bX^{\natural} \right\Vert_{\F}}\bH\right)\bV^{\top}\ba_{i}\right)^{2} - \left(\ba_{i}^{\top}\left(\bX^{\natural} + t_{2}\frac{\sigma_{r}^{2}\left(\bX^{\natural}\right)}{\left\Vert \bX^{\natural} \right\Vert_{\F}}\bH\right)\bV^{\top}\ba_{i}\right)^{2} \right\vert   \\
& \le 6 n  \cdot \sqrt{6 n } \cdot \frac{\sigma_{r}^{2}\left(\bX^{\natural}\right)}{\left\Vert \bX^{\natural} \right\Vert_{\F}}\epsilon \cdot 2\sqrt{6 n} \cdot \frac{25}{24}\big\Vert\bX^{\natural}\big\Vert  + 2\cdot 6 n \cdot \frac{\sigma_{r}^{2}\left(\bX^{\natural}\right)}{\left\Vert \bX^{\natural} \right\Vert_{\F}}\epsilon \cdot 12 n  \cdot \frac{25}{24}\big\Vert\bX^{\natural}\big\Vert\\
& \le 225 \epsilon n^{2} \frac{\sigma_{r}^{2}\left(\bX^{\natural}\right)}{\left\Vert \bX^{\natural} \right\Vert_{\F}}  \big\Vert\bX^{\natural}\big\Vert
 \le \frac{1}{24} \sigma_{r}^{2}\left(\bX^{\natural}\right),
\end{align*}
as long as $\epsilon = \frac{1}{5400 n^{2}} \cdot \frac{\left\Vert\bX^{\natural}\right\Vert_{\F}}{\left\Vert\bX^{\natural}\right\Vert}$. The cardinality of this $\epsilon$-net will be $\frac{1/24}{\epsilon} \le c n^{2} \cdot \frac{\left\Vert\bX^{\natural}\right\Vert}{\left\Vert\bX^{\natural}\right\Vert_{\F}}$.


Therefore, when $m\ge c \frac{ \big\Vert \bX^{\natural} \big\Vert_{\F}^{4}}{\sigma_{r}^{4}\left(\bX^{\natural}\right) } n r \log{\left(n \kappa\right)}$ with some large enough constant $c$, for all matrices $\bV$ and $\bX$ such that $\left\Vert\bX-\bX^{\natural}\right\Vert_{\F}\le \frac{1}{24}\frac{\sigma_{r}^{2}\left(\bX^{\natural}\right)}{\left\Vert \bX^{\natural} \right\Vert_{\F}}$, we have 
 \begin{equation}\label{equ_hessian_q_finallowerbound}
 q\left(\bV, \bH, t\right) \ge   \big\Vert \bX^{\natural} \big\Vert_{\F}^{2}  + 2 \big\Vert \bX^{\natural \top}\bV\big\Vert_{\F}^{2}  +  2 \mathrm{Tr}\left(\bX^{\natural \top}\bV\bX^{\natural \top}\bV\right)  + 1.246  \sigma_{r}^{2}\left(\bX^{\natural}\right) ,
\end{equation}
with probability at least $1 - e^{-c_{1}nr \log{\left(n \kappa\right)}} - m e^{-1.5n}$.

\subsection{Finishing the Proof}

Combining \eqref{equ_hessian_loose_secondterm} and \eqref{equ_hessian_q_finallowerbound}, we can prove
\begin{align*}
\mathrm{vec}\left(\bV\right)^{\top} \nabla^{2}f(\bX) \mathrm{vec}\left(\bV\right)
& \ge \left\Vert \bX^{\natural} \right\Vert_{\F}^{2} + 2 \left\Vert \bX^{\natural \top}\bV\right\Vert_{\F}^{2}  +  2 \mathrm{Tr}\left(\bX^{\natural \top}\bV\bX^{\natural \top}\bV\right)  + 1.246  \sigma_{r}^{2}\left(\bX^{\natural}\right) \\
& \quad - \frac{1}{m}\sum_{i=1}^{m} \left\Vert \ba_{i}^{\top}\bX^{\natural}\right\Vert_{2}^{2} \left\Vert \ba_{i}^{\top}\bV\right\Vert_{2}^{2}\\
& \ge  \left\Vert \bX^{\natural} \right\Vert_{\F}^{2} + 2 \left\Vert \bX^{\natural \top}\bV\right\Vert_{\F}^{2}  +  2 \mathrm{Tr}\left(\bX^{\natural \top}\bV\bX^{\natural \top}\bV\right)  + 1.246  \sigma_{r}^{2}\left(\bX^{\natural}\right)\\
& \quad - \left\Vert \bX^{\natural}\right\Vert_{\F}^{2} - 2 \left\Vert \bX^{\natural \top}\bV\right\Vert_{\F}^{2} -  \frac{1}{24} \sigma_{r}^{2}\left(\bX^{\natural}\right)  \\
& \ge  2 \mathrm{Tr}\left(\bX^{\natural \top}\bV\bX^{\natural \top}\bV\right)  + 1.204  \sigma_{r}^{2}\left(\bX^{\natural}\right)
\end{align*}
as claimed. 

\section{Proof of Lemma~\ref{lemma:lemma_induction}}\label{proof_lemma_induction}
We first note that
\begin{align}
\left\Vert \bX_{t+1}\bQ_{t+1} - \bX^{\natural} \right\Vert_{\F}^{2} 
& \le \left\Vert \bX_{t+1}\bQ_{t} - \bX^{\natural} \right\Vert_{\F}^{2} \label{equ_induc_to_groundtruth_qt}\\
& = \left\Vert \left( \bX_{t} - \mu \nabla f\left(\bX_{t}\right)\right) \bQ_{t} - \bX^{\natural} \right\Vert_{\F}^{2} \nonumber\\
& =  \left\Vert  \bX_{t} \bQ_{t} - \mu \nabla f\left(\bX_{t} \bQ_{t}\right) - \bX^{\natural} \right\Vert_{\F}^{2} \label{equ_induc_to_groundtruth_gradeq}\\
& = \left\Vert \bx_{t} - \bx^{\natural} - \mu \cdot  \mathrm{vec}\left(\nabla f\left(\bX_{t} \bQ_{t}\right) - \nabla f\left(\bX^{\natural}\right)\right)   \right\Vert_{2}^{2}, \label{equ_induc_to_groundtruth_zerogd} 
\end{align}
where we write 
\begin{equation*}
\bx_{t} := \mathrm{vec}\left(\bX_{t}\bQ_{t} \right) \quad \text{and}\quad \bx^{\natural} := \mathrm{vec}\big(\bX^{\natural}\big).
\end{equation*}
Here, \eqref{equ_induc_to_groundtruth_qt} follows from the definition of $\bQ_{t+1}$ (see \eqref{eq:defn-Qt}),  \eqref{equ_induc_to_groundtruth_gradeq} holds owing to the identity $\nabla f\left(\bX_{t}\right) \bQ_{t} = \nabla f\left(\bX_{t} \bQ_{t}\right)$ for $\bQ_t\in\mathcal{O}^{r\times r}$, and \eqref{equ_induc_to_groundtruth_zerogd} arises from the fact that $\nabla f\left(\bX^{\natural}\right) = \boldsymbol{0}$.  Let
\begin{equation*}
\bX_{t}(\tau) = \bX^{\natural} + \tau \big( \bX_{t}\bQ_{t} - \bX^{\natural} \big), 
\end{equation*}
where $\tau\in[0,1]$. Then, by the fundamental theorem of calculus for vector-valued functions \cite{lang1993real}, 
\begin{align}
\mbox{RHS of \eqref{equ_induc_to_groundtruth_zerogd} }
& = \left\Vert \bx_{t} - \bx^{\natural} - \mu \cdot  \int_{0}^{1} \nabla^{2} f\left(\bX_{t}(\tau) \right) \left(\bx_{t} - \bx^{\natural}\right) d\tau  \right\Vert_{2}^{2}\label{equ_induc_to_groundtruth_meanvalue}\\
	& = \left\Vert \left(\bI - \mu \cdot  \int_{0}^{1} \nabla^{2} f\left(\bX_{t}(\tau) \right) \mathrm{d}\tau\right) \left(\bx_{t} - \bx^{\natural}\right)  \right\Vert_{2}^{2}\nonumber\\
	& = \left(\bx_{t} - \bx^{\natural}\right)^{\top} \left(\bI - \mu \cdot  \int_{0}^{1} \nabla^{2} f\left(\bX_{t}(\tau) \right) \mathrm{d}\tau\right)^{2}   \left(\bx_{t} - \bx^{\natural}\right)\nonumber\\
	& =  \left\Vert \bx_{t} - \bx^{\natural} \right\Vert_{2}^{2} - 2\mu  \cdot  \left(\bx_{t} - \bx^{\natural}\right)^{\top} \left(\int_{0}^{1} \nabla^{2} f\left(\bX_{t}(\tau) \right) \mathrm{d}\tau\right) \left(\bx_{t} - \bx^{\natural}\right) \nonumber\\
	& \qquad + \mu^{2}  \cdot \left(\bx_{t} - \bx^{\natural}\right)^{\top} \left(\int_{0}^{1} \nabla^{2} f\left(\bX_{t}(\tau) \right) \mathrm{d}\tau\right)^{2} \left(\bx_{t} - \bx^{\natural}\right) \nonumber\\
	&\le \left\Vert \bx_{t} - \bx^{\natural} \right\Vert_{2}^{2} - 2\mu  \cdot  \left(\bx_{t} - \bx^{\natural}\right)^{\top} \left(\int_{0}^{1} \nabla^{2} f\left(\bX_{t}(\tau) \right) \mathrm{d}\tau\right) \left(\bx_{t} - \bx^{\natural}\right) \nonumber\\
	& \qquad + \mu^{2}  \cdot  \left\Vert\int_{0}^{1} \nabla^{2} f\left(\bX_{t}(\tau) \right) \mathrm{d}\tau\right\Vert^{2} \left\Vert\bx_{t} - \bx^{\natural}\right\Vert_{2}^{2} .\label{eq:decom_bound}
\end{align}

It is easy to verify that $\bX_{t}(\tau)$ satisfies \eqref{eq:def_RIC} for any $\tau\in[0,1]$, since
\begin{equation*}
\left\Vert \bX_{t}(\tau) - \bX^{\natural} \right\Vert_{\F} = \tau\left\Vert  \bX_{t}\bQ_{t} - \bX^{\natural} \right\Vert_{\F}  \le \frac{1}{24} \frac{\sigma_{r}^{2}\left(\bX^{\natural}\right)}{\left\Vert \bX^{\natural} \right\Vert_{\F}},
\end{equation*}
and
\begin{equation*}
\max_{1\le l\le m} \left\Vert \ba_{l}^{\top}\left(\bX_{t}(\tau) - \bX^{\natural}\right) \right\Vert_{2} = \tau \cdot\max_{1\le l\le m} \left\Vert \ba_{l}^{\top}\left(\bX_{t}\bQ_{t} - \bX^{\natural}\right) \right\Vert_{2} \leq \frac{1}{24} \sqrt{\log{n}}  \cdot \frac{\sigma_{r}^{2}\left(\bX^{\natural}\right)}{\left\Vert \bX^{\natural} \right\Vert_{\F}}.
\end{equation*}
 Lemma~\ref{lemma_restrict_concen_hessian_neighbor} then implies that 
\begin{equation*}
	\left(\bx_{t} - \bx^{\natural}\right)^{\top} \left(\int_{0}^{1} \nabla^{2} f\left(\bX_{t}(\tau) \right) \mathrm{d}\tau\right) \left(\bx_{t} - \bx^{\natural}\right) \ge 1.026  \sigma_{r}^{2}\left(\bX^{\natural}\right) \left\Vert \bx_{t} - \bx^{\natural} \right\Vert_{2}^{2},
\end{equation*}
and 
\begin{equation*}
	\left\Vert\int_{0}^{1} \nabla^{2} f\left(\bX_{t}(\tau) \right) \mathrm{d}\tau\right\Vert \le 1.5 \sigma_{r}^{2}\left(\bX^{\natural}\right) \log{n}    + 6\big\Vert\bX^{\natural}\big\Vert_{\F}^{2}.
\end{equation*}
Substituting the above two inequalities into \eqref{equ_induc_to_groundtruth_zerogd} and \eqref{eq:decom_bound} gives
\begin{align*}
&\left\Vert \bX_{t+1}\bQ_{t+1} - \bX^{\natural} \right\Vert_{\F}^{2}\\
&\le  \left\Vert \bx_{t} - \bx^{\natural} \right\Vert_{2}^{2} - 2\mu \cdot  1.026  \sigma_{r}^{2}\left(\bX^{\natural}\right) \left\Vert \bx_{t} - \bx^{\natural} \right\Vert_{2}^{2} + \mu^{2}  \cdot  \left( 1.5  \sigma_{r}^{2}\left(\bX^{\natural}\right) \log{n}  + 6\left\Vert\bX^{\natural}\right\Vert_{\F}^{2} \right)^{2} \left\Vert \bx_{t} - \bx^{\natural} \right\Vert_{2}^{2}\\
& = \left[ 1 - 2.052 \sigma_{r}^{2}\left(\bX^{\natural}\right) \mu  + \left(1.5 \sigma_{r}^{2}\left(\bX^{\natural}\right) \log{n}    + 6\left\Vert\bX^{\natural}\right\Vert_{\F}^{2}\right)^{2} \mu^{2} \right] \left\Vert \bX_{t}\bQ_{t} - \bX^{\natural} \right\Vert_{\F}^{2}\\
& \le  \left(1 - 1.026 \sigma_{r}^{2}\left(\bX^{\natural}\right) \mu \right) \left\Vert \bX_{t}\bQ_{t} - \bX^{\natural} \right\Vert_{\F}^{2},
\end{align*}
with the proviso that $\mu \le \frac{1.026 \sigma_{r}^{2}\left(\bX^{\natural}\right)}{\left(1.5 \sigma_{r}^{2}\left(\bX^{\natural}\right) \log{n}    + 6\left\Vert\bX^{\natural}\right\Vert_{\F}^{2}\right)^{2}}$. This allows us to conclude that
\begin{equation*}
\big\Vert \bX_{t+1}\bQ_{t+1} - \bX^{\natural} \big\Vert_{\F} \le \left(1 - 0.513 \sigma_{r}^{2}\big(\bX^{\natural}\big) \mu \right) \big\Vert \bX_{t}\bQ_{t} - \bX^{\natural} \big\Vert_{\F}.
\end{equation*}

\section{Proof of Lemma~\ref{lemma:proximity}}\label{proof:lemma_proximity}
Recognizing that
\begin{align*}
 \left\Vert \bX_{t+1}\bQ_{t+1} - \bX_{t+1}^{(l)}\bR_{t+1}^{(l)} \right\Vert_{\F} 
 & \le  \left\Vert \bX_{t+1}\bQ_{t+1} - \bX_{t+1}^{(l)} \bR_{t}^{(l)} \bQ_{t}^{\top} \bQ_{t+1} \right\Vert_{\F} \\ &=  \left\Vert \bX_{t+1}  - \bX_{t+1}^{(l)} \bR_{t}^{(l)} \bQ_{t}^{\top}  \right\Vert_{\F}
 = \left\Vert \bX_{t+1}\bQ_{t}  - \bX_{t+1}^{(l)} \bR_{t}^{(l)}   \right\Vert_{\F},
\end{align*}
we will focus on bounding $\big\Vert \bX_{t+1}\bQ_{t}  - \bX_{t+1}^{(l)} \bR_{t}^{(l)}   \big\Vert_{\F}$. Since 
\begin{align*}
& \bX_{t+1}\bQ_{t}  - \bX_{t+1}^{(l)} \bR_{t}^{(l)}
 = \left( \bX_{t} - \mu \nabla f\left(\bX_{t}\right) \right) \bQ_{t}  - \left( \bX_{t}^{(l)} - \mu \nabla f^{(l)} \left(\bX_{t}^{(l)} \right)  \right) \bR_{t}^{(l)}\\
& =  \bX_{t}\bQ_{t} - \bX_{t}^{(l)}\bR_{t}^{(l)} - \mu \nabla f\left(\bX_{t}\right) \bQ_{t}  + \mu \nabla f^{(l)}\left(\bX_{t}^{(l)}\right)\bR_{t}^{(l)}\\
& =  \bX_{t}\bQ_{t} - \bX_{t}^{(l)}\bR_{t}^{(l)} - \mu  \frac{1}{m} \sum_{i=1}^{m } \left( \left\Vert \ba_{i}^{\top}\bX_{t} \right\Vert_{2}^{2} - y_{i} \right) \ba_{i}\ba_{i}^{\top}\bX_{t} \bQ_{t}  \\
& \quad + \mu   \frac{1}{m} \sum_{i=1}^{m} \left( \left\Vert \ba_{i}^{\top}\bX_{t}^{(l)} \right\Vert_{2}^{2} - y_{i} \right) \ba_{i}\ba_{i}^{\top}\bX_{t}^{(l)}  \bR_{t}^{(l)} - \mu  \frac{1}{m}  \left( \left\Vert \ba_{l}^{\top}\bX_{t}^{(l)} \right\Vert_{2}^{2} - y_{l} \right) \ba_{l}\ba_{l}^{\top}\bX_{t}^{(l)}  \bR_{t}^{(l)}\\
& = \underbrace{\bX_{t}\bQ_{t} - \bX_{t}^{(l)}\bR_{t}^{(l)} - \mu  \nabla f\left(\bX_{t}\bQ_{t}\right) +  \mu  \nabla f\left(\bX_{t}^{(l)}\bR_{t}^{(l)}\right)}_{:=\bS_{t,1}^{(l)}}  - \underbrace{ \mu  \frac{1}{m}  \left( \left\Vert \ba_{l}^{\top}\bX_{t}^{(l)} \right\Vert_{2}^{2} - y_{l} \right) \ba_{l}\ba_{l}^{\top}\bX_{t}^{(l)}  \bR_{t}^{(l)}}_{:=\bS_{t,2}^{(l)}},
\end{align*}
we aim to control $\big\Vert \bS_{t,1}^{(l)}\big\Vert_{\F}$ and $\big\Vert \bS_{t,2}^{(l)} \big\Vert_{\F}$ separately.

We first bound the term $\big\Vert \bS_{t,2}^{(l)}\big\Vert_{\F}$, which is easier to handle. Observe that by  Cauchy-Schwarz, 
\begin{align}
\left\vert \left\Vert \ba_{l}^{\top}\bX_{t}^{(l)} \right\Vert_{2}^{2} - y_{l} \right\vert & = \left\vert \ba_{l}^{\top}\left(\bX_{t}^{(l)}\bR_{t}^{(l)} - \bX^{\natural}\right) \left(\bX_{t}^{(l)}\bR_{t}^{(l)} + \bX^{\natural}\right)^{\top}\ba_{l}  \right\vert  \nonumber\\
& \le \left\Vert \ba_{l}^{\top}\left(\bX_{t}^{(l)}\bR_{t}^{(l)} - \bX^{\natural}\right)\right\Vert_{2} \left\Vert \ba_{l}^{\top}\left(\bX_{t}^{(l)}\bR_{t}^{(l)} + \bX^{\natural}\right)\right\Vert_{2}  .\label{equ_inductive_leaveapprox_s2_dis_cauchy}
\end{align}
The first term in \eqref{equ_inductive_leaveapprox_s2_dis_cauchy} can be bounded by
\begin{align}
&\left\Vert \ba_{l}^{\top}\left(\bX_{t}^{(l)}\bR_{t}^{(l)} - \bX^{\natural}\right)\right\Vert_{2} \nonumber\\
 & \leq\left\Vert \ba_{l}^{\top}\left(\bX_{t}^{(l)}\bR_{t}^{(l)} - \bX_{t}\bQ_{t}\right)\right\Vert_{2} + \left\Vert \ba_{l}^{\top}\left(\bX_{t}\bQ_{t} - \bX^{\natural}\right)\right\Vert_{2} \nonumber\\
&  \leq \sqrt{6n}  \left\Vert \bX_{t}^{(l)}\bR_{t}^{(l)} - \bX_{t}\bQ_{t} \right\Vert + C_{2} \left(1 - 0.5 \sigma_{r}^{2}\left(\bX^{\natural}\right) \mu \right)^{t} \sqrt{\log{n}}  \cdot \frac{\sigma_{r}^{2}\left(\bX^{\natural}\right)}{\left\Vert \bX^{\natural} \right\Vert_{\F}} \nonumber \\
& \le  \sqrt{6n}  C_{3} \left(1 - 0.5 \sigma_{r}^{2}\left(\bX^{\natural}\right) \mu \right)^{t} \sqrt{\frac{\log{n}}{n}} \cdot  \frac{\sigma_{r}^{2}\left(\bX^{\natural}\right)}{\kappa \left\Vert \bX^{\natural} \right\Vert_{\F}} +  C_{2} \left(1 - 0.5 \sigma_{r}^{2}\left(\bX^{\natural}\right) \mu \right)^{t}\sqrt{\log{n}}  \cdot \frac{\sigma_{r}^{2}\left(\bX^{\natural}\right)}{\left\Vert \bX^{\natural} \right\Vert_{\F}}    \nonumber\\
&\le  (\sqrt{6}C_3+C_2) \left(1 - 0.5 \sigma_{r}^{2}\left(\bX^{\natural}\right) \mu \right)^{t} \sqrt{\log{n}}  \frac{\sigma_{r}^{2}\left(\bX^{\natural}\right)}{\left\Vert \bX^{\natural} \right\Vert_{\F}},\label{eq:incoherence_diff}
\end{align}
where we have used the triangle inequality, Lemma~\ref{lemma_a_sqrtroot_concen}, as well as the induction hypotheses \eqref{equ_inductive_att_c2} and \eqref{equ_inductive_att_c3}. Similarly, the second term in \eqref{equ_inductive_leaveapprox_s2_dis_cauchy} can be bounded as
\begin{align}
 \left\Vert \ba_{l}^{\top}\left(\bX_{t}^{(l)}\bR_{t}^{(l)} + \bX^{\natural}\right)\right\Vert_{2} & \leq \left\Vert \ba_{l}^{\top}\left(\bX_{t}^{(l)}\bR_{t}^{(l)} - \bX^{\natural}\right)\right\Vert_{2} + 2 \left\Vert \ba_{l}^{\top}\bX^{\natural}\right\Vert_{2}    \nonumber\\
&\le \left(\sqrt{6}C_{3}+C_{2}\right) \sqrt{\log{n}}  \frac{\sigma_{r}^{2}\left(\bX^{\natural}\right)}{\left\Vert \bX^{\natural} \right\Vert_{\F}}  + 11.72  \sqrt{\log{n}} \left\Vert\bX^{\natural}\right\Vert_{\F}   \nonumber\\
&\le \left(\sqrt{6}C_{3}+C_{2} + 11.72  \right)\sqrt{\log{n}} \big\Vert\bX^{\natural}\big\Vert_{\F}, \label{eq:incoherence_sum} 
\end{align}
where we have used \eqref{eq:incoherence_diff}, Lemma~\ref{lemma_a_log_tight_concen}, and $\sigma_r^2\left(\bX^{\natural}\right)\leq  \big\Vert\bX^{\natural}\big\Vert_{\F}^2$. Similarly, we can also obtain
\begin{align*}
 \left\Vert \ba_{l}^{\top}\bX_{t}^{(l)} \right\Vert_{2} \le  \left(\sqrt{6}C_{3}+C_{2} + 5.86  \right)\sqrt{\log{n}} \big\Vert\bX^{\natural}\big\Vert_{\F}.
\end{align*}
Substituting \eqref{eq:incoherence_diff} and \eqref{eq:incoherence_sum} into \eqref{equ_inductive_leaveapprox_s2_dis_cauchy}, and using the above inequality, we get
\begin{align}
\left\Vert \bS_{t,2}^{(l)}\right\Vert_{\F}
& = \mu   \frac{1}{m} \cdot \left\vert \left\Vert \ba_{l}^{\top}\bX_{t}^{(l)} \right\Vert_{2}^{2} - y_{l} \right\vert   \cdot \left\Vert  \ba_{l}\ba_{l}^{\top}\bX_{t}^{(l)} \right\Vert_{\F} \nonumber\\
& \le C_4^2 \left(1 - 0.5 \sigma_{r}^{2}\left(\bX^{\natural}\right) \mu \right)^{t} \cdot \mu   \frac{1}{m} \cdot \sigma_{r}^{2}\left(\bX^{\natural}\right)   \log{n}  \cdot \left\Vert\ba_{l}\right\Vert_{2} \left\Vert \ba_{l}^{\top}\bX_{t}^{(l)} \right\Vert_{2} \nonumber\\
& \le  \sqrt{6}C_4^3 \left(1 - 0.5 \sigma_{r}^{2}\left(\bX^{\natural}\right) \mu \right)^{t} \cdot \mu   \frac{1}{m} \cdot \sigma_{r}^{2}\left(\bX^{\natural}\right)   \log{n}  \cdot \sqrt{n} \big\Vert\bX^{\natural} \big\Vert_{\F} \sqrt{\log{n}} \nonumber\\
& =  \sqrt{6}C_4^3 \left(1 - 0.5 \sigma_{r}^{2}\left(\bX^{\natural}\right) \mu \right)^{t} \cdot \mu \frac{\sqrt{n}\cdot\left(\log{n}\right)^{3/2}}{m} \sigma_{r}^{2}\left(\bX^{\natural}\right)\big\Vert\bX^{\natural} \big\Vert_{\F} ,\label{equ_induc_leaveapp_s2}
\end{align}
where $C_4: = \sqrt{6}C_{3}+C_{2} + 11.72 $.

Next, we turn to $\left\Vert \bS_{t,1}^{(l)}\right\Vert_{\F}$. By defining 
	$$\bs_{t,1}^{(l)} = \mathrm{vec}\big(\bS_{t,1}^{(l)}\big), \quad \bx_{t} = \mathrm{vec}\left(\bX_{t}\bQ_{t}\right), \quad \text{and} \quad \bx_{t}^{(l)} = \mathrm{vec}\big( \bX_{t}^{(l)}\bR_{t}^{(l)} \big),$$ 
	we can write
\begin{align*}
\bs_{t,1}^{(l)} & = \bx_{t} - \bx_{t}^{(l)} - \mu \cdot  \mathrm{vec}\left(  \nabla f\left(\bX_{t}\bQ_{t}\right) -  \nabla f (\bX_{t}^{(l)}\bR_{t}^{(l)} ) \right)\\
	& = \bx_{t} - \bx_{t}^{(l)} - \mu \cdot \int_{0}^{1} \nabla^{2} f\left(\bX_{t}^{(l)}(\tau) \right) \left(\bx_{t} - \bx_{t}^{(l)}\right) \mathrm{d}\tau\\
	& = \left(\bI - \mu \cdot \int_{0}^{1} \nabla^{2} f\left(\bX_{t}^{(l)}(\tau) \right) \mathrm{d}\tau\right)  \left( \bx_{t} - \bx_{t}^{(l)}\right).
\end{align*}
Here, the second line follows from the fundamental theorem of calculus for vector-valued functions \cite{lang1993real}, where
\begin{equation}
\bX_{t}^{(l)}(\tau) = \bX_{t}^{(l)}\bR_{t}^{(l)} + \tau \left( \bX_{t}\bQ_{t} - \bX_{t}^{(l)}\bR_{t}^{(l)} \right), 
\end{equation}
for $\tau\in[0,1]$. Using very similar algebra as in Appendix~\ref{proof_lemma_induction}, we obtain
\begin{align}
	\left\Vert \bS_{t,1}^{(l)}\right\Vert_{\F}^{2} & \le \left\Vert \bx_{t} - \bx_{t}^{(l)}\right\Vert_{2}^{2} + \mu^{2} \left\Vert  \int_{0}^{1} \nabla^{2} f\left(\bX_{t}^{(l)}(\tau) \right) \mathrm{d}\tau \right\Vert^{2} \left\Vert \bx_{t} - \bx_{t}^{(l)}\right\Vert_{2}^{2} \nonumber\\
	& \quad - 2\mu \cdot \left( \bx_{t} - \bx_{t}^{(l)}\right)^{\top} \left(\int_{0}^{1} \nabla^{2} f\left(\bX_{t}^{(l)}(\tau) \right) \mathrm{d}\tau\right)  \left( \bx_{t} - \bx_{t}^{(l)}\right).\label{bound_St1}
\end{align}

It is easy to verify that for all $\tau\in[0,1]$,
\begin{align}
\left\Vert \bX_{t}^{(l)}(\tau) - \bX^{\natural} \right\Vert_{\F} 
& = \left\Vert \left(1-\tau\right)\left( \bX_{t}^{(l)}\bR_{t}^{(l)} - \bX_{t}\bQ_{t} \right) + \bX_{t}\bQ_{t} - \bX^{\natural}  \right\Vert_{\F} \nonumber\\
& \le  ( 1-\tau) \left\Vert  \bX_{t}^{(l)}\bR_{t}^{(l)} - \bX_{t}\bQ_{t}  \right\Vert_{\F} + \left\Vert \bX_{t}\bQ_{t} - \bX^{\natural}  \right\Vert_{\F} \nonumber\\
&\le C_{3} \sqrt{\frac{\log{n}}{n}} \cdot  \frac{\sigma_{r}^{2}\left(\bX^{\natural}\right)}{\kappa \left\Vert \bX^{\natural} \right\Vert_{\F}} +  C_{1} \frac{\sigma_{r}^{2}\left(\bX^{\natural}\right)}{\left\Vert \bX^{\natural} \right\Vert_{\F}} \label{equ_loo_induc_s1_invokehypo}\\
& \le \left(C_{3} \sqrt{\frac{\log{n}}{n}} + C_{1}\right) \frac{\sigma_{r}^{2}\left(\bX^{\natural}\right)}{\left\Vert \bX^{\natural} \right\Vert_{\F}} \le \frac{1}{24} \frac{\sigma_{r}^{2}\left(\bX^{\natural}\right)}{\left\Vert \bX^{\natural} \right\Vert_{\F}}, \label{equ_loo_induc_s1_setc}
\end{align}
where \eqref{equ_loo_induc_s1_invokehypo} follows from the induction hypotheses \eqref{equ_inductive_att_c1} and \eqref{equ_inductive_att_c3}, and \eqref{equ_loo_induc_s1_setc} follows  as long as $C_{1}+C_{3} \le \frac{1}{24}$.  Further, for all $1\leq l \leq m$, by the induction hypothesis \eqref{equ_inductive_att_c3} and \eqref{equ_inductive_att_c2},
\begin{align*}
  \left\Vert \ba_{l}^{\top}\left(\bX_{t}^{(l)}(\tau) - \bX^{\natural}\right) \right\Vert_{2} 
& \le ( 1-\tau)  \left\Vert \ba_{l}^{\top}\left( \bX_{t}^{(l)}\bR_{t}^{(l)} - \bX_{t}\bQ_{t}  \right) \right\Vert_{2}  +  \left\Vert \ba_{l}^{\top}\left( \bX_{t}\bQ_{t} - \bX^{\natural} \right) \right\Vert_{2} \\
& \le   \left\Vert \ba_{l}\right\Vert_{2} \left\Vert \bX_{t}^{(l)}\bR_{t}^{(l)} - \bX_{t}\bQ_{t} \right\Vert + C_{2} \sqrt{\log{n}}  \cdot \frac{\sigma_{r}^{2}\left(\bX^{\natural}\right)}{\left\Vert \bX^{\natural} \right\Vert_{\F}}\\
& \le \sqrt{6n} C_{3} \sqrt{\frac{\log{n}}{n}} \cdot  \frac{\sigma_{r}^{2}\left(\bX^{\natural}\right)}{\kappa \left\Vert \bX^{\natural} \right\Vert_{\F}}+ C_{2} \sqrt{\log{n}}  \cdot \frac{\sigma_{r}^{2}\left(\bX^{\natural}\right)}{\left\Vert \bX^{\natural} \right\Vert_{\F}}\\
&\le \left(\sqrt{6}C_{3} + C_{2}\right) \sqrt{\log{n}}  \cdot \frac{\sigma_{r}^{2}\left(\bX^{\natural}\right)}{\left\Vert \bX^{\natural} \right\Vert_{\F}} \le \frac{1}{24} \sqrt{\log{n}}  \cdot \frac{\sigma_{r}^{2}\left(\bX^{\natural}\right)}{\left\Vert \bX^{\natural} \right\Vert_{\F}},
\end{align*}
as long as $ \sqrt{6}C_{3} + C_{2} \le \frac{1}{24}$. Therefore, Lemma~\ref{lemma_restrict_concen_hessian_neighbor} holds for $\bX_{t}^{(l)}(\tau)$, and similar to Appendix~\ref{proof_lemma_induction}, \eqref{bound_St1} can be further bounded by  
\begin{equation}\label{equ_induc_leaveapp_s1}
\left\Vert \bS_{t,1}^{(l)}\right\Vert_{\F} \le \left(1 - 0.513 \sigma_{r}^{2}\left(\bX^{\natural}\right) \mu \right) \left\Vert \bX_{t}\bQ_{t} - \bX_{t}^{(l)}\bR_{t}^{(l)}\right\Vert_{\F}
\end{equation}
as long as  $\mu \le \frac{1.026 \sigma_{r}^{2}\left(\bX^{\natural}\right)}{\left(1.5 \sigma_{r}^{2}\left(\bX^{\natural}\right) \log{n}    + 6\left\Vert\bX^{\natural}\right\Vert_{\F}^{2}\right)^{2}}$. Consequently, combining \eqref{equ_induc_leaveapp_s2} and \eqref{equ_induc_leaveapp_s1}, we can get 
\begin{align}
\left\Vert \bX_{t+1}\bQ_{t+1} - \bX_{t+1}^{(l)}\bR_{t+1}^{(l)}\right\Vert_{\F} 
& \le \left\Vert \bS_{t,1}^{(l)}\right\Vert_{\F} + \left\Vert \bS_{t,2}^{(l)}\right\Vert_{\F} \nonumber\\
& \le \left(1 - 0.513 \sigma_{r}^{2}\left(\bX^{\natural}\right) \mu \right) \left\Vert \bX_{t}\bQ_{t} - \bX_{t}^{(l)}\bR_{t}^{(l)}\right\Vert_{\F} \nonumber\\
& \quad +  \sqrt{6}C_4^3 \left(1 - 0.5 \sigma_{r}^{2}\left(\bX^{\natural}\right) \mu \right)^{t} \cdot \mu \frac{\sqrt{n}\cdot\left(\log{n}\right)^{3/2}}{m} \sigma_{r}^{2}\left(\bX^{\natural}\right) \big\Vert\bX^{\natural} \big\Vert_{\F}  \nonumber\\
& \le C_{3} \left(1 - 0.5 \sigma_{r}^{2}\left(\bX^{\natural}\right) \mu \right)^{t+1} \sqrt{\frac{\log{n}}{n}}  \cdot \frac{\sigma_{r}^{2}\left(\bX^{\natural}\right)}{\kappa \left\Vert\bX^{\natural} \right\Vert_{\F}}, \label{equ_loo_induc_finalbound_setm}
\end{align}
where \eqref{equ_loo_induc_finalbound_setm} follows from the induction hypothesis \eqref{equ_inductive_att_c3}, as long as $m\ge c \kappa \frac{\left\Vert\bX^{\natural} \right\Vert_{\F}^{2}}{\sigma_{r}^{2}\left(\bX^{\natural}\right)} n\log{n}$ for some large enough constant $c>0$.

\section{Proof of Lemma~\ref{lemma:incoherence_induction}} \label{proof_incoherence_induction}

For any $1\leq l\leq m$, by the statistical independence of $\ba_l$ and $\bX_{t+1}^{(l)}$ and by Lemma~\ref{lemma_a_log_tight_concen}, we have
\begin{equation*}
\left\Vert \ba_{l}^{\top}\left( \bX_{t+1}^{(l)}\bQ_{t+1}^{(l)} - \bX^{\natural}\right) \right\Vert_{2}  \leq 5.86 \sqrt{\log{n}} \left\Vert  \bX_{t+1}^{(l)} \bQ_{t+1}^{(l)} - \bX^{\natural}  \right\Vert_{\F} .
\end{equation*}
Since following Lemma~\ref{lemma:lemma_induction}, 
\begin{align*}
\left\Vert \bX_{t+1}\bQ_{t+1} - \bX^{\natural} \right\Vert \left\Vert  \bX^{\natural} \right\Vert 
& \le \left\Vert \bX_{t+1}\bQ_{t+1} - \bX^{\natural} \right\Vert_{\F} \left\Vert  \bX^{\natural} \right\Vert \\
& \le C_{1} \left(1 - 0.513 \sigma_{r}^{2}\left(\bX^{\natural}\right) \mu \right)^{t+1}  \cdot  \frac{\sigma_{r}^{2}\left(\bX^{\natural}\right)}{\left\Vert \bX^{\natural} \right\Vert_{\F}} \cdot \left\Vert  \bX^{\natural} \right\Vert \\
& \le \frac{1}{2} \sigma_{r}^{2}\left(\bX^{\natural}\right),
\end{align*}
as long as $C_{1} \le \frac{1}{2}$, and following Lemma~\ref{lemma:proximity},
\begin{align*}
\left\Vert \bX_{t+1}\bQ_{t+1} -  \bX_{t+1}^{(l)}\bR_{t+1}^{(l)}  \right\Vert \left\Vert \bX^{\natural} \right\Vert
& \le \left\Vert \bX_{t+1}\bQ_{t+1} -  \bX_{t+1}^{(l)}\bR_{t+1}^{(l)}  \right\Vert_{\F} \left\Vert \bX^{\natural} \right\Vert\\
& \le C_{3} \left(1 - 0.5 \sigma_{r}^{2}\left(\bX^{\natural}\right) \mu \right)^{t+1}  \sqrt{\frac{\log{n}}{n}} \cdot  \frac{\sigma_{r}^{2}\left(\bX^{\natural}\right)}{\kappa \left\Vert \bX^{\natural} \right\Vert_{\F}}  \cdot \left\Vert \bX^{\natural} \right\Vert\\
& \le \frac{1}{4} \sigma_{r}^{2}\left(\bX^{\natural}\right),
\end{align*}
as long as $C_{3} \le \frac{1}{4}$, we can invoke Lemma 37 in \cite{ma2017implicit} and get 
\begin{align*}
\left\Vert \bX_{t+1}\bQ_{t+1}  - \bX_{t+1}^{(l)}\bQ_{t+1}^{(l)}  \right\Vert_{\F} \le 5 \kappa \left\Vert \bX_{t+1}\bQ_{t+1}  - \bX_{t+1}^{(l)}\bR_{t+1}^{(l)}  \right\Vert_{\F}.
\end{align*}

Further, by the triangle inequality, Lemma~\ref{lemma_a_sqrtroot_concen}, Lemma~\ref{lemma:proximity} and Lemma~\ref{lemma:lemma_induction}, we can deduce that 
\begin{align}
& \left\Vert \ba_{l}^{\top} \left( \bX_{t+1}\bQ_{t+1} - \bX^{\natural} \right) \right\Vert_{2} 
 \le  \left\Vert \ba_{l}^{\top} \left( \bX_{t+1}\bQ_{t+1} - \bX_{t+1}^{(l)}\bQ_{t+1}^{(l)} \right) \right\Vert_{2} +  \left\Vert \ba_{l}^{\top}\left( \bX_{t+1}^{(l)}\bQ_{t+1}^{(l)} - \bX^{\natural}\right) \right\Vert_{2} \nonumber\\
& \le  \left\Vert \ba_{l}\right\Vert_{2} \left\Vert \bX_{t+1}\bQ_{t+1} - \bX_{t+1}^{(l)}\bQ_{t+1}^{(l)}  \right\Vert + 5.86 \sqrt{\log{n}} \left\Vert  \bX_{t+1}^{(l)} \bQ_{t+1}^{(l)} - \bX^{\natural}  \right\Vert_{\F} \nonumber\\
& \le \sqrt{6n}   \left\Vert \bX_{t+1}\bQ_{t+1} - \bX_{t+1}^{(l)}\bQ_{t+1}^{(l)}  \right\Vert + 5.86 \sqrt{\log{n}}  \left\Vert \bX_{t+1}\bQ_{t+1} -  \bX_{t+1}^{(l)}\bQ_{t+1}^{(l)} \right\Vert_{\F} \nonumber\\
&\quad + 5.86 \sqrt{\log{n}} \left\Vert \bX_{t+1}\bQ_{t+1} - \bX^{\natural} \right\Vert_{\F} \nonumber\\
& \le \left(\sqrt{6n} + 5.86 \sqrt{\log{n}} \right)  \left\Vert \bX_{t+1}\bQ_{t+1} -  \bX_{t+1}^{(l)}\bQ_{t+1}^{(l)} \right\Vert_{\F} + 5.86 \sqrt{\log{n}}  \left\Vert \bX_{t+1}\bQ_{t+1} - \bX^{\natural} \right\Vert_{\F} \nonumber\\
& \le 5 \left(\sqrt{6n} + 5.86 \sqrt{\log{n}} \right) \kappa  \left\Vert \bX_{t+1}\bQ_{t+1} -  \bX_{t+1}^{(l)}\bR_{t+1}^{(l)} \right\Vert_{\F} + 5.86 \sqrt{\log{n}}  \left\Vert \bX_{t+1}\bQ_{t+1} - \bX^{\natural} \right\Vert_{\F} \nonumber\\
& \le 5 \left( \sqrt{6n}  + 5.86 \sqrt{\log{n}}  \right)  \kappa  \cdot C_{3} \left(1 - 0.5 \sigma_{r}^{2}\left(\bX^{\natural}\right) \mu \right)^{t+1}  \sqrt{\frac{\log{n}}{n}} \cdot  \frac{\sigma_{r}^{2}\left(\bX^{\natural}\right)}{\kappa \left\Vert \bX^{\natural} \right\Vert_{\F}} \nonumber\\
& \quad +  5.86 \sqrt{\log{n}}  \cdot C_{1} \left(1 - 0.513 \sigma_{r}^{2}\left(\bX^{\natural}\right) \mu \right)^{t+1}  \cdot  \frac{\sigma_{r}^{2}\left(\bX^{\natural}\right)}{\left\Vert \bX^{\natural} \right\Vert_{\F}} \nonumber \\
& \le \left( 5\sqrt{6}C_{3} +  5.86 C_{1} + 29.3 C_{3}  \sqrt{\frac{\log{n}}{n}} \right)   \left(1 - 0.5 \sigma_{r}^{2}\left(\bX^{\natural}\right) \mu \right)^{t+1}  \sqrt{\log{n}} \cdot \frac{\sigma_{r}^{2}\left(\bX^{\natural}\right)}{\left\Vert \bX^{\natural} \right\Vert_{\F}} \nonumber\\
&\le  C_{2} \left(1 - 0.5 \sigma_{r}^{2}\left(\bX^{\natural}\right) \mu \right)^{t+1} \sqrt{\log{n}} \cdot  \frac{\sigma_{r}^{2}\left(\bX^{\natural}\right)}{\left\Vert \bX^{\natural} \right\Vert_{\F}},\nonumber
\end{align}
where the last line follows as long as $5\sqrt{6}C_{3} +  5.86 C_{1} + 29.3 C_{3} \le  C_{2}$. The proof is then finished by applying the union bound for all $1\leq l\leq m$.



\section{Proof of Lemma~\ref{lemma:initialization}}\label{proof_lemma_initialization}

Define 
\begin{align*}
\boldsymbol{\Sigma}_{0} &= \mathrm{diag}\left\{ \lambda_{1}\left(\bY\right), \lambda_{2}\left(\bY\right), \cdots , \lambda_{r}\left(\bY\right) \right\} = \boldsymbol{\Lambda}_{0} + \lambda\bI \\ \boldsymbol{\Sigma}_{0}^{(l)} &= \mathrm{diag} \left\{ \lambda_{1}\left(\bY^{(l)}\right), \lambda_{2}\left(\bY^{(l)}\right),  \cdots , \lambda_{r}\left(\bY^{(l)}\right) \right\} = \boldsymbol{\Lambda}_{0}^{(l)} + \lambda^{(l)}\bI, \quad 1\leq l\leq m,
\end{align*}
then by definition we have $\bY\bZ_{0} = \bZ_{0}\boldsymbol{\Sigma}_{0}$, $\bY^{(l)}\bZ_{0}^{(l)} = \bZ_{0}^{(l)}\boldsymbol{\Sigma}_{0}^{(l)}$, and 
\begin{equation}\label{eq:one_sample}
\boldsymbol{\Sigma}_{0}\bZ_{0}^{\top}\bZ_{0}^{(l)} - \bZ_{0}^{\top}\bZ_{0}^{(l)}\boldsymbol{\Sigma}_{0}^{(l)} = \frac{1}{2m} y_{l}\bZ_{0}^{\top}\ba_{l}\ba_{l}^{\top}\bZ_{0}^{(l)}.
\end{equation} 
Moreover, let $\bZ_{0,c}$ and $\bZ_{0,c}^{(l)}$ be the complement matrices of $\bZ_{0}$ and $\bZ_{0}^{(l)}$, respectively, such that both $\left[ \bZ_{0}, \bZ_{0,c} \right]$ and $\left[ \bZ_{0}^{(l)}, \bZ_{0,c}^{(l)} \right]$ are orthonormal matrices. Below we will prove the induction hypotheses \eqref{eq:induction_hyp} in the base case when $t=0$ one by one.

\subsection{Proof of \eqref{equ_inductive_att_c1} }


From Lemma~\ref{lemma_fronorm_lowbound}, we have 
\begin{align}
\left\Vert \bX_{0}\bQ_{0} - \bX^{\natural}\right\Vert_{\F}  & \le \frac{1}{\sqrt{2\left(\sqrt{2}-1\right)} \sigma_{r} \left( \bX^{\natural} \right) } \left\Vert \bX_{0}\bX_{0}^{\top} - \bX^{\natural}\bX^{\natural \top} \right\Vert_{\F} \nonumber\\
& = \frac{1}{\sqrt{2\left(\sqrt{2}-1\right)} \sigma_{r} \left( \bX^{\natural} \right) } \left\Vert \bZ_{0}\boldsymbol{\Lambda}_{0}\bZ_{0}^{\top} - \bX^{\natural}\bX^{\natural \top} \right\Vert_{\F}  \nonumber\\
& \le \frac{\sqrt{r}}{\sqrt{2\left(\sqrt{2}-1\right)} \sigma_{r} \left( \bX^{\natural} \right) } \left\Vert \bZ_{0}\boldsymbol{\Sigma}_{0}\bZ_{0}^{\top} - \bX^{\natural}\bX^{\natural \top} - \lambda \bZ_{0}\bZ_{0}^{\top} \right\Vert . \label{eq_fronorm_lowerbound}
\end{align}
The last term in \eqref{eq_fronorm_lowerbound} can be further bounded as
\begin{align}
& \left\Vert \bZ_{0}\boldsymbol{\Sigma}_{0}\bZ_{0}^{\top} - \bX^{\natural}\bX^{\natural \top} - \lambda \bZ_{0}\bZ_{0}^{\top} \right\Vert \nonumber \\
& \le   \left\Vert \bY -  \frac{1}{2} \Vert \bX^{\natural}  \Vert_{\F}^{2}\bI - \bX^{\natural}\bX^{\natural \top}  \right\Vert  
+ \left\Vert \bZ_{0}\boldsymbol{\Sigma}_{0}\bZ_{0}^{\top}  -\bY  + \frac{1}{2} \Vert \bX^{\natural} \Vert_{\F}^{2}\bZ_{0,c}\bZ_{0,c}^{\top}  \right\Vert + \left\Vert \frac{1}{2} \Vert \bX^{\natural}  \Vert_{\F}^{2}\bZ_{0}\bZ_{0}^{\top} - \lambda \bZ_{0}\bZ_{0}^{\top} \right\Vert   \nonumber\\
& \le    \delta \left\Vert \bX^{\natural}\right\Vert_{\F}^{2} +  \delta \left\Vert \bX^{\natural}\right\Vert_{\F}^{2} + \delta \left\Vert \bX^{\natural} \right\Vert_{\F}^{2}    = 3\delta  \left\Vert \bX^{\natural}\right\Vert_{\F}^{2}, \label{equ_initial_spectral_method_dis}
\end{align}
where \eqref{equ_initial_spectral_method_dis} follows from
\begin{equation*} 
\left\Vert \bY - \mathbb{E}[\bY]\right\Vert = \left\Vert \bY -  \frac{1}{2}\left\Vert \bX^{\natural}\right\Vert_{\F}^{2}\bI - \bX^{\natural}\bX^{\natural \top} \right\Vert \le \delta \left\Vert \bX^{\natural} \right\Vert_{\F}^{2}
\end{equation*}
via Lemma~\ref{lemma_intial_weighted_mat_concen}, the Weyl's inequality, and 
\begin{equation*}
\left\vert \lambda - \mathbb{E}\left[\lambda\right] \right\vert = \left\vert \lambda - \frac{1}{2} \left\Vert\bX^{\natural}\right\Vert_{\F}^{2} \right\vert \le \delta \left\Vert \bX^{\natural}\right\Vert_{\F}^{2}
\end{equation*}
via Lemma~\ref{lemma_aaT_operator_concentration}. Plugging \eqref{equ_initial_spectral_method_dis} into \eqref{eq_fronorm_lowerbound}, we have
\begin{align}
\left\Vert \bX_{0}\bQ_{0} - \bX^{\natural}\right\Vert_{\F}  & \le \frac{3}{\sqrt{2\left(\sqrt{2}-1\right)}} \cdot \frac{\delta \sqrt{r} \left\Vert \bX^{\natural}\right\Vert_{\F}^{2}}{\sigma_{r}\left( \bX^{\natural} \right) },\nonumber
\end{align}
Setting $\delta = c \frac{\sigma_{r}^{3}\left( \bX^{\natural} \right) }{\sqrt{r}\left\Vert \bX^{\natural} \right\Vert_{\F}^{3}}$ for a sufficiently small constant $c$, i.e. $m\gtrsim \frac{\left\Vert \bX^{\natural} \right\Vert_{\F}^{6}}{\sigma_{r}^{6}\left( \bX^{\natural} \right) } nr\log n $, we get $\left\Vert \bX_{0}\bQ_{0} - \bX^{\natural}\right\Vert_{\F} \allowbreak \le C_1 \frac{\sigma_{r}^{2}\left(\bX^{\natural}\right)}{\left\Vert\bX^{\natural}\right\Vert_{\F}}$. Following similar procedures, we can also show $ \left\Vert \bX_{0}^{(l)}\bQ_{0}^{(l)} - \bX^{\natural}\right\Vert_{\F} \allowbreak \le C_1 \frac{\sigma_{r}^{2}\left(\bX^{\natural}\right)}{\left\Vert \bX^{\natural}\right\Vert_{\F}}$.

\subsection{Proof of  \eqref{equ_inductive_att_c3} }

Following Weyl's inequality, by \eqref{equ_inductive_att_c1}, we have 
	$$\left\vert \sigma_{i}\left(\bX_{0}\right) - \sigma_{i}\left(\bX^{\natural}\right) \right\vert \le C_1 \frac{\sigma_{r}^{2}\left(\bX^{\natural}\right)}{\left\Vert\bX^{\natural}\right\Vert_{\F}},$$ 
	and similarly, $\left\vert \sigma_{i}\left(\bX_{0}^{(l)}\right) - \sigma_{i}\left(\bX^{\natural}\right) \right\vert \le C_1 \frac{\sigma_{r}^{2}\left(\bX^{\natural}\right)}{\left\Vert \bX^{\natural} \right\Vert_{\F}}$, for $i=1,\cdots,r$. 
Combined with Lemma~\ref{lemma_fronorm_lowbound}, there exists some constant $c$ such that
\begin{align}
&\left\Vert \bX_{0}\bQ_{0} - \bX_{0}^{(l)}\bR_{0}^{(l)}\right\Vert_{\F}  \le \frac{1}{\sqrt{2\left(\sqrt{2}-1\right)} \sigma_{r} \left( \bX_{0} \right) } \left\Vert \bX_{0}\bX_{0}^{\top} - \bX_{0}^{(l)}\bX_{0}^{(l) \top} \right\Vert_{\F} \nonumber\\
& \le  \frac{c}{\sigma_{r} \left( \bX^{\natural} \right) } \left\Vert \bX_{0}\bX_{0}^{\top} - \bX_{0}^{(l)}\bX_{0}^{(l) \top} \right\Vert_{\F} \nonumber\\
& = \frac{c}{\sigma_{r} \left( \bX^{\natural} \right) } \left\Vert \bZ_{0}\boldsymbol{\Lambda}_{0}\bZ_{0}^{\top} - \bZ_{0}^{(l)}\boldsymbol{\Lambda}_{0}^{(l)}\bZ_{0}^{(l) \top} \right\Vert_{\F} \nonumber\\
& = \frac{c}{\sigma_{r} \left( \bX^{\natural} \right) } \left\Vert \bZ_{0}\boldsymbol{\Sigma}_{0}\bZ_{0}^{\top} - \bZ_{0}^{(l)}\boldsymbol{\Sigma}_{0}^{(l)}\bZ_{0}^{(l) \top} - \lambda\bZ_{0}\bZ_{0}^{\top} + \lambda^{(l)} \bZ_{0}^{(l)}\bZ_{0}^{(l) \top}  \right\Vert_{\F} \nonumber\\
&\le \frac{c}{\sigma_{r} \left( \bX^{\natural} \right) } \left\Vert \bZ_{0}\boldsymbol{\Sigma}_{0}\bZ_{0}^{\top} - \bZ_{0}^{(l)}\boldsymbol{\Sigma}_{0}^{(l)}\bZ_{0}^{(l) \top} \right\Vert_{\F} +  \frac{c}{\sigma_{r} \left( \bX^{\natural} \right) } \left\Vert  \lambda \bZ_{0}\bZ_{0}^{\top} -  \lambda^{(l)}\bZ_{0}^{(l)}\bZ_{0}^{(l) \top}  \right\Vert_{\F}. \label{equ_initial_leaveoneout_approx_bound}
\end{align}
We will bound each term in \eqref{equ_initial_leaveoneout_approx_bound}, respectively. For the first term, we have
\begin{align}
&\left\Vert \bZ_{0}\boldsymbol{\Sigma}_{0}\bZ_{0}^{\top} - \bZ_{0}^{(l)}\boldsymbol{\Sigma}_{0}^{(l)}\bZ_{0}^{(l) \top} \right\Vert_{\F}  = \left\Vert \left[ \bZ_{0}\boldsymbol{\Sigma}_{0}\bZ_{0}^{\top}\bZ_{0}^{(l)} - \bZ_{0}^{(l)}\boldsymbol{\Sigma}_{0}^{(l)}, \bZ_{0}\boldsymbol{\Sigma}_{0}\bZ_{0}^{\top} \bZ_{0,c}^{(l)}  \right]  \right\Vert_{\F} \nonumber\\
& \le  \left\Vert \bZ_{0}\boldsymbol{\Sigma}_{0}\bZ_{0}^{\top} \bZ_{0}^{(l)} - \bZ_{0}^{(l)}\boldsymbol{\Sigma}_{0}^{(l)} \right\Vert_{\F} + \left\Vert  \bZ_{0}\boldsymbol{\Sigma}_{0}\bZ_{0}^{\top} \bZ_{0,c}^{(l)} \right\Vert_{\F} \nonumber\\
& \le \left\Vert \bZ_{0}\boldsymbol{\Sigma}_{0}\bZ_{0}^{\top} \bZ_{0}^{(l)} - \bZ_{0}\bZ_{0}^{\top}\bZ_{0}^{(l)}\boldsymbol{\Sigma}_{0}^{(l)} \right\Vert_{\F} + \left\Vert \bZ_{0}\bZ_{0}^{\top}\bZ_{0}^{(l)}\boldsymbol{\Sigma}_{0}^{(l)} - \bZ_{0}^{(l)}\boldsymbol{\Sigma}_{0}^{(l)} \right\Vert_{\F} + \left\Vert\bY\right\Vert \left\Vert \bZ_{0}^{\top} \bZ_{0,c}^{(l)} \right\Vert_{\F} \nonumber\\
& \le  \left\Vert \bZ_{0} \cdot  \frac{1}{2m} y_{l}\bZ_{0}^{\top}\ba_{l}\ba_{l}^{\top}\bZ_{0}^{(l)} \right\Vert_{\F} + \left\Vert  \bZ_{0}\bZ_{0}^{\top} - \bZ_{0}^{(l)}\bZ_{0}^{(l) \top} \right\Vert_{\F} \big\Vert\bY^{(l)}\big\Vert  + \left\Vert\bY\right\Vert \left\Vert \bZ_{0}^{\top} \bZ_{0,c}^{(l)} \right\Vert_{\F} ,  \label{equ_initial_leave_approx_first_angle} 
\end{align}
where the last line follows from \eqref{eq:one_sample}. Note that the first term in \eqref{equ_initial_leave_approx_first_angle} can be bounded as
\begin{align}
\left\Vert \bZ_{0} \cdot  \frac{1}{2m} y_{l}\bZ_{0}^{\top}\ba_{l}\ba_{l}^{\top}\bZ_{0}^{(l)} \right\Vert_{\F} & \leq  \frac{1}{2m} \left\Vert \ba_{l}^{\top}\bX^{\natural} \right\Vert_{2}^{2}  \left\Vert \ba_{l}^{\top}\bZ_{0}^{(l)} \right\Vert_{2} \left\Vert \ba_{l}^{\top}\bZ_{0} \right\Vert_{2}  \nonumber \\
& \lesssim  \frac{\sqrt{n} \cdot \left(\log{n}\right)^{3/2} \cdot \sqrt{r}}{m} \big\Vert\bX^{\natural}\big\Vert_{\F}^{2} , \label{eq:bound_first_term}
 \end{align}
which follows Lemma~\ref{lemma_a_sqrtroot_concen} and Lemma~\ref{lemma_a_log_tight_concen}. The second term in \eqref{equ_initial_leave_approx_first_angle} can be bounded as
 \begin{align*}
 \left\Vert   \bZ_{0}\bZ_{0}^{\top} - \bZ_{0}^{(l)}\bZ_{0}^{(l) \top}  \right\Vert_{\F} & =    \left\Vert \bZ_{0}\left(\bZ_{0} - \bZ_{0}^{(l)}\bT_{0}^{(l)}\right)^{\top} +  \left(\bZ_{0} - \bZ_{0}^{(l)}\bT_{0}^{(l)}\right) \left(\bZ_{0}^{(l)}\bT_{0}^{(l)}\right)^{\top}\right\Vert_{\F} \\
 & \leq 2   \left\Vert \bZ_{0} - \bZ_{0}^{(l)}\bT_{0}^{(l)} \right\Vert_{\F}    \leq 2\sqrt{2}  \left\Vert \bZ_{0}^{\top} \bZ_{0,c}^{(l)} \right\Vert_{\F} 
 \end{align*}
where $\bT_{t}^{(l)}  = \argmin_{\bP\in\mathcal{O}^{r\times r} } \left\Vert \bZ_{t} - \bZ_{t}^{(l)}\bP\right\Vert_{\F}$, and the last line follows from the fact $\left\Vert \bZ_{0} - \bZ_{0}^{(l)}\bT_{0}^{(l)} \right\Vert_{\F}  \le \sqrt{2}  \left\Vert \bZ_{0}^{\top} \bZ_{0,c}^{(l)} \right\Vert_{\F} $ \cite{yu2014useful}. Putting this together with the third term in \eqref{equ_initial_leave_approx_first_angle}, we have
\begin{align}
\left\Vert  \bZ_{0}\bZ_{0}^{\top} - \bZ_{0}^{(l)}\bZ_{0}^{(l) \top} \right\Vert_{\F} \big\Vert\bY^{(l)}\big\Vert & + \left\Vert\bY\right\Vert \left\Vert \bZ_{0}^{\top} \bZ_{0,c}^{(l)} \right\Vert_{\F}  \le  \left(2\sqrt{2}\left\Vert\bY^{(l)}\right\Vert + \left\Vert\bY \right\Vert \right)  \left\Vert \bZ_{0}^{\top} \bZ_{0,c}^{(l)} \right\Vert_{\F} \nonumber \\
& \lesssim \left\Vert\bX^{\natural}\right\Vert_{\F}^{2} \frac{\left\Vert \left(\frac{1}{m}y_{l}\ba_{l}\ba_{l}^{\top}\right)\bZ_{0}^{(l)} \right\Vert_{\F}}{\sigma_{r}^{2}\left(\bX^{\natural}\right)} \label{equ_intial_leave_approx_first_daviskahan}\\
& \lesssim \frac{ \left\Vert \ba_{l}^{\top}\bX^{\natural} \right\Vert_{2}^{2}  \left\Vert \ba_{l}^{\top}\bZ_{0}^{(l)} \right\Vert_{2} \left\Vert \ba_{l} \right\Vert_{2}}{m}\frac{\left\Vert\bX^{\natural}\right\Vert_{\F}^{2}}{\sigma_{r}^{2}\left(\bX^{\natural}\right)} \nonumber\\
& \lesssim \frac{\sqrt{n} \cdot \left(\log{n}\right)^{3/2} \cdot \sqrt{r}}{m}\frac{ \left\Vert\bX^{\natural}\right\Vert_{\F}^{4}}{\sigma_{r}^{2}\left(\bX^{\natural}\right)}, \label{equ_initial_leave_approx_first_conbound}
\end{align}
where \eqref{equ_intial_leave_approx_first_daviskahan} follows from Lemma~\ref{lemma_intial_weighted_mat_concen} and the Davis-Kahan $\sin\Theta$ theorem \cite{davis1970rotation}, and \eqref{equ_initial_leave_approx_first_conbound} follows from Lemma~\ref{lemma_a_sqrtroot_concen} and Lemma~\ref{lemma_a_log_tight_concen}. 

For the second term in \eqref{equ_initial_leaveoneout_approx_bound}, we have
\begin{align}
\left\Vert  \lambda \bZ_{0}\bZ_{0}^{\top} -  \lambda^{(l)}\bZ_{0}^{(l)}\bZ_{0}^{(l) \top}  \right\Vert_{\F}
& = \left\Vert \lambda \bZ_{0}\bZ_{0}^{\top} -  \lambda \bZ_{0}^{(l)}\bZ_{0}^{(l) \top} + \lambda \bZ_{0}^{(l)}\bZ_{0}^{(l) \top} -  \lambda^{(l)}\bZ_{0}^{(l)}\bZ_{0}^{(l) \top}  \right\Vert_{\F} \nonumber\\
& \le  \lambda \cdot   \left\Vert \bZ_{0}\bZ_{0}^{\top} -  \bZ_{0}^{(l)}\bZ_{0}^{(l) \top}  \right\Vert_{\F}  +  \left|\lambda -  \lambda^{(l)}\right| \cdot \left\Vert  \bZ_{0}^{(l)}\bZ_{0}^{(l) \top} \right\Vert_{\F} \nonumber\\
& \lesssim \frac{\sqrt{n} \cdot \left(\log{n}\right)^{3/2} \cdot \sqrt{r}}{m}\frac{ \left\Vert\bX^{\natural}\right\Vert_{\F}^{4}}{\sigma_{r}^{2}\left(\bX^{\natural}\right)} +  \frac{y_l }{2m} \sqrt{r}   \label{equ_intial_leave_approx_second_lambda}\\
& \lesssim \frac{\sqrt{n} \cdot \left(\log{n}\right)^{3/2} \cdot \sqrt{r}}{m}\frac{ \left\Vert\bX^{\natural}\right\Vert_{\F}^{4}}{\sigma_{r}^{2}\left(\bX^{\natural}\right)}+  \frac{\sqrt{r} \cdot \log{n}}{m} \left\Vert \bX^{\natural} \right\Vert_{\F}^{2} ,\label{equ_intial_leave_approx_second}
\end{align}
where the first term of \eqref{equ_intial_leave_approx_second_lambda} is bounded similarly as \eqref{equ_initial_leave_approx_first_conbound}, and \eqref{equ_intial_leave_approx_second} follows from Lemma~\ref{lemma_a_log_tight_concen}. Combining \eqref{eq:bound_first_term}, \eqref{equ_initial_leave_approx_first_conbound}, and \eqref{equ_intial_leave_approx_second}, we obtain
\begin{align*}
\left\Vert \bX_{0}\bQ_{0} - \bX_{0}^{(l)}\bR_{0}^{(l)}\right\Vert_{\F} 
\lesssim  \frac{\sqrt{n} \cdot \left(\log{n}\right)^{3/2} \cdot \sqrt{r}}{m}\frac{ \left\Vert\bX^{\natural}\right\Vert_{\F}^{4}}{\sigma_{r}^{3}\left(\bX^{\natural}\right)}
\lesssim \sqrt{\frac{\log{n}}{n}}  \cdot  \frac{\sigma_{r}^{2}\left(\bX^{\natural}\right)}{\kappa \left\Vert\bX^{\natural}\right\Vert_{\F}},
\end{align*}
where the last inequality holds as long as $m \gtrsim \kappa \frac{\left\Vert\bX^{\natural}\right\Vert_{\F}^{5}}{\sigma_{r}^{5}\left(\bX^{\natural}\right)} n\sqrt{r}\log{n} = O(nr^3\log n)$.

\subsection{Proof of \eqref{equ_inductive_att_c2}}

Since from \eqref{equ_inductive_att_c1} and \eqref{equ_inductive_att_c3},
\begin{align*}
\left\Vert  \bX_{0}\bQ_{0} - \bX^{\natural} \right\Vert \left\Vert \bX^{\natural} \right\Vert 
\le \left\Vert  \bX_{0}\bQ_{0} - \bX^{\natural} \right\Vert_{\F} \left\Vert \bX^{\natural} \right\Vert 
\lesssim \sigma_{r}^{2}\left(\bX^{\natural}\right),
\end{align*}
and for every $1\leq l\leq m$, 
\begin{align*}
\left\Vert \bX_{0}\bQ_{0} - \bX_{0}^{(l)}\bR_{0}^{(l)} \right\Vert \left\Vert \bX^{\natural} \right\Vert 
\le \left\Vert \bX_{0}\bQ_{0} - \bX_{0}^{(l)}\bR_{0}^{(l)}  \right\Vert_{\F} \left\Vert \bX^{\natural} \right\Vert   
\lesssim \sigma_{r}^{2}\left(\bX^{\natural}\right),
\end{align*}
with proper constants, following Lemma 37 in \cite{ma2017implicit}, we have 
\begin{align*}
\left\Vert \bX_{0}\bQ_{0}  - \bX_{0}^{(l)}\bQ_{0}^{(l)}  \right\Vert_{\F} \le 5 \kappa \left\Vert \bX_{0}\bQ_{0}  - \bX_{0}^{(l)}\bR_{0}^{(l)}  \right\Vert_{\F},
\end{align*}
which implies that for every $1\leq l\leq m$ we can get
\begin{align*}
\left\Vert \bX_{0}^{(l)}\bQ_{0}^{(l)} - \bX^{\natural} \right\Vert_{\F}
& \le  \left\Vert \bX_{0}\bQ_{0} - \bX_{0}^{(l)}\bQ_{0}^{(l)} \right\Vert_{\F}  + \left\Vert \bX_{0}\bQ_{0} - \bX^{\natural} \right\Vert_{\F}\\
&  \lesssim \kappa \left\Vert \bX_{0}\bQ_{0}  - \bX_{0}^{(l)}\bR_{0}^{(l)}  \right\Vert_{\F} + \left\Vert \bX_{0}\bQ_{0} - \bX^{\natural} \right\Vert_{\F} \\
& \lesssim \kappa \sqrt{\frac{\log{n}}{n}}  \cdot  \frac{\sigma_{r}^{2}\left(\bX^{\natural}\right)}{\kappa \left\Vert\bX^{\natural}\right\Vert_{\F}} + \frac{\sigma_{r}^{2}\left(\bX^{\natural}\right)}{\left\Vert\bX^{\natural}\right\Vert_{\F}} \\
&\lesssim  \frac{\sigma_{r}^{2}\left(\bX^{\natural}\right)}{\left\Vert\bX^{\natural}\right\Vert_{\F}}.
\end{align*}
This further gives
\begin{align}
&\max_{1\le l\le m} \left\Vert \ba_{l}^{\top} \left( \bX_{0}\bQ_{0} - \bX^{\natural} \right) \right\Vert_{2} \nonumber \\
& \le \max_{1\le l\le m} \left\Vert \ba_{l}^{\top} \left( \bX_{0}\bQ_{0} - \bX_{0}^{(l)}\bQ_{0}^{(l)} \right) \right\Vert_{2} + \max_{1\le l\le m}\left\Vert \ba_{l}^{\top}\left( \bX_{0}^{(l)}\bQ_{0}^{(l)} - \bX^{\natural}\right) \right\Vert_{2} \nonumber\\
& \le \max_{1\le l\le m} \left\Vert \ba_{l}\right\Vert_{2} \left\Vert \bX_{0}\bQ_{0} - \bX_{0}^{(l)}\bQ_{0}^{(l)}  \right\Vert + \max_{1\le l\le m}\left\Vert \ba_{l}^{\top}\left( \bX_{0}^{(l)}\bQ_{0}^{(l)} - \bX^{\natural}\right) \right\Vert_{2} \nonumber \\
& \lesssim \sqrt{n} \cdot  \max_{1\le l\le m} \left\Vert \bX_{0}\bQ_{0} - \bX_{0}^{(l)}\bQ_{0}^{(l)}  \right\Vert + \sqrt{\log{n}} \cdot  \max_{1\le l\le m}\left\Vert \bX_{0}^{(l)}\bQ_{0}^{(l)} - \bX^{\natural} \right\Vert_{2} \label{equ_intial_twoabound} \\
& \lesssim \sqrt{n} \cdot \kappa \max_{1\le l\le m} \left\Vert \bX_{0}\bQ_{0} - \bX_{0}^{(l)}\bR_{0}^{(l)}  \right\Vert + \sqrt{\log{n}} \cdot  \max_{1\le l\le m}\left\Vert \bX_{0}^{(l)}\bQ_{0}^{(l)} - \bX^{\natural} \right\Vert_{2} \nonumber \\
& \lesssim  \sqrt{n} \cdot \kappa \sqrt{\frac{\log{n}}{n}}  \cdot  \frac{\sigma_{r}^{2}\left(\bX^{\natural}\right)}{\kappa \left\Vert\bX^{\natural}\right\Vert_{\F}}  +  \sqrt{\log{n}} \cdot \frac{\sigma_{r}^{2}\left(\bX^{\natural}\right)}{\left\Vert\bX^{\natural}\right\Vert_{\F}} \label{equ_intial_abound_logtight}\\
	& \lesssim  \sqrt{\log{n}} \cdot \frac{\sigma_{r}^{2}\left(\bX^{\natural}\right)}{\left\Vert\bX^{\natural}\right\Vert_{\F}}, \nonumber
\end{align}
where \eqref{equ_intial_twoabound} follows from Lemma~\ref{lemma_a_sqrtroot_concen} and Lemma~\ref{lemma_a_log_tight_concen}, and \eqref{equ_intial_abound_logtight} follows from \eqref{equ_inductive_att_c3}.


\subsection{Finishing the Proof}
The proof of Lemma~\ref{lemma:initialization} is now complete by appropriately adjusting the constants.



\bibliographystyle{IEEEtran} 
\bibliography{bibfileNonconvex}

\end{document}